\pdfoutput=1

\newif \iffull  \fulltrue
\newif \ifdraft \drafttrue

\input{tex.flags}

\iffull
\documentclass{article}
\usepackage{fullpage}
\usepackage{lmodern}
\else
\documentclass[preprint]{sigplanconf}
\fi
\usepackage{natbib}

\usepackage[T1]{fontenc}                                               %
\usepackage{microtype}                                                 %

\usepackage[utf8]{inputenc}
\usepackage[english]{babel}

\usepackage{amsmath}
\usepackage{amssymb}
\usepackage{wasysym}
\usepackage{stmaryrd}
\usepackage{amsthm}

\usepackage{mathtools}

\usepackage{xspace}
\usepackage{xcolor}

\usepackage{paralist}

\usepackage{listings}

\usepackage[scaled]{beramono}
\newcommand\Small{\fontsize{8.2pt}{8.4pt}\selectfont}
\newcommand*\LSTfont{\Small\ttfamily\SetTracking{encoding=*}{-60}\lsstyle}

\newcommand{\lstt}[1]{\mbox{\LSTfont #1}}

\usepackage{hyperref}
\definecolor{DarkGreen}{rgb}{0.1,0.5,0.1}
\definecolor{DarkRed}{rgb}{0.5,0.1,0.1}
\definecolor{DarkBlue}{rgb}{0.1,0.1,0.5}
\usepackage[]{hyperref}
\hypersetup{
    unicode=false,          
    pdftoolbar=true,        
    pdfmenubar=true,        
    pdffitwindow=false,      
    pdftitle={},    
    pdfauthor={}
    pdfsubject={},   
    pdfnewwindow=true,      
    pdfkeywords={keywords}, 
    colorlinks=true,       
    linkcolor=DarkRed,          
    citecolor=DarkGreen,        
    filecolor=DarkRed,      
    urlcolor=DarkBlue,          
}
\usepackage{cleveref}

\def\Rextp{\ensuremath{\mathbb{S}}\xspace}

\def\R{\mathbb{R}}
\def\N{\mathbb{N}}

\def\st{\mid}

\def\umodels{\models^U}

\newcommand{\vsub}{\ensuremath{\models}}

\newcommand{\vd}{\ensuremath{\vdash}}
\newcommand{\vs}{\ensuremath{\vdash_{\mathcal{S}}}}

\newcommand{\req}[1]{{#1}_{\Box\uparrow}}

\def\pair#1#2{{\langle #1 , #2 \rangle}}

\def\amp{\mathrel{\binampersand}}

\def\lin{\multimap}
\def\lol{\multimap}
\def\scale{{!}}

\def\evtoign#1{ \hookrightarrow} 

\def\blet{\mathop{\bf let}}

\def\bin{\mathop{\bf in}}

\def\bcase{\mathop{\bf case}\nolimits}
\def\bof{\mathop{\bf of}}

\def\myrel#1\over#2{\mathrel{\mathop{\kern0pt#2}\limits_{#1}}}

\def\breturn{\mathop{\mathbf{return}}}
\def\bclub{\mathop{\mathbf{club}}}

\def\lb{\llbracket}
\def\rb{\rrbracket}

\def\tless{\sqsubseteq}

\def\esens{\ensuremath{\hat{R}}}

\def\smax#1#2{\ensuremath{\mathbf{max}({#1},{#2})}}
\def\ssup#1#2{\ensuremath{\mathbf{sup}({#1},{#2})}}
\def\scase#1#2#3#4{\ensuremath{\mathbf{case}(#1 , {#2}, {#3}, {#4})}}
\def\sbound{\ensuremath{\mathbf{bound}}}
\def\club#1{\ensuremath{\mathbf{club}{\{ #1 \}}}}
\def\csens#1#2#3{\ensuremath{(#1; #2; #3)}}
\def\cred{\mapsto}

\newcommand{\Fuzz}{{{\em Fuzz}}\xspace}
\newcommand{\DFuzz}{{{\em DFuzz}}\xspace}
\newcommand{\EDFuzz}{{{\em EDFuzz}}\xspace}
\newcommand{\UDFuzz}{{{\em UDFuzz}}\xspace}

\newcommand{\dlPCF}{$\mathsf{d}\ell\mathsf{PCF}$}
\newcommand{\PCF}{$\mathsf{PCF}$}

\newcommand{\cass}[3]{{#1} :_{[#2]} {#3}}
\newcommand{\cempty}[2]{{#1} :_{\Box} {#2}}

\def\contextone{\Gamma}
\def\contexttwo{\Delta}
\def\contextthree{\Sigma}

\def\ctxo{\contextone}
\def\ctxw{\contexttwo}
\def\ctxt{\contextthree}

\def\ctxskel{\ctxo^\bullet}
\def\ctxwskel{\ctxw^\bullet}

\newcommand{\iconstraintone}{\Phi}
\newcommand{\iconstrainttwo}{\Psi}

\def\cso{\iconstraintone}
\def\csw{\iconstrainttwo}

\newcommand{\icontextone}{\phi}
\newcommand{\icontexttwo}{\psi}

\def\tctxo{\icontextone}
\def\tctxw{\icontexttwo}

\def\ectxo{\mathcal{E}}
\def\cctxo{\mathcal{C}}

\def\typeone{\sigma}
\def\typetwo{\tau}
\def\typethree{\mu}

\def\tyo{\typeone}
\def\tyw{\typetwo}
\def\tyt{\typethree}
\def\rrb{\rrparenthesis}
\def\llb{\llparenthesis}

\def\sizeivarone{i}

\def\sensitermone{R}
\def\sizeitermone{S}
\def\tless{\sqsubseteq}

\def\varone{x}

\def\termone{e}
\def\tmo{\termone}

\def\realone{\mathrm{r}}
\def\natone{\mathrm{n}}

\newcommand\cclean[1]{{#1}^\bullet}
\newcommand\sclean[1]{\overline{#1}}

\newcommand{\rname}[1]{\quad #1}
\newcommand{\brname}[1]{(#1)\xspace}

\newcommand{\struletens}{\ensuremath{\brname{\tless.\otimes}}}
\newcommand{\struleamp}{\ensuremath{\brname{\tless.\amp}}}
\newcommand{\strulelin}{\ensuremath{\brname{\tless.\lin}}}
\newcommand{\struleforall}{\ensuremath{\brname{\tless.\forall}}}

\newcommand{\rulesubl}{\ensuremath{\brname{\mathrm{\tless.L}}}}
\newcommand{\rulesubr}{\ensuremath{\brname{\mathrm{\tless.R}}}}
\newcommand{\rulevar}{\ensuremath{\brname{\mathrm{Var}}}}
\newcommand{\ruleconst}{\ensuremath{\brname{\mathrm{Const}}}}
\newcommand{\ruleconstR}{\ensuremath{\brname{\mathrm{Const}_\R}}}
\newcommand{\ruleconstN}{\ensuremath{\brname{\mathrm{Const}_\N}}}
\newcommand{\rulefix}{\ensuremath{\brname{\mathrm{Fix}}}}
\newcommand{\ruleitens}{\ensuremath{\brname{\otimes I}}}
\newcommand{\ruleetens}{\ensuremath{\brname{\otimes E}}}
\newcommand{\ruleiamp}{\ensuremath{\brname{\amp I}}}
\newcommand{\ruleeamp}{\ensuremath{\brname{\amp E}}}
\newcommand{\ruleiapp}{\ensuremath{\brname{\multimap I}}}
\newcommand{\ruleeapp}{\ensuremath{\brname{\multimap E}}}
\newcommand{\ruleitapp}{\ensuremath{\brname{\forall I}}}
\newcommand{\ruleetapp}{\ensuremath{\brname{\forall E}}}

\newcommand{\rulenats}{\ensuremath{\brname{\mathrm{S}~I}}}
\newcommand{\rulenate}{\ensuremath{\brname{\N~E}}}

\newcommand{\semu}[1]{\llbracket #1\rrbracket}

\newcommand{\dom}{{\text{dom}}}

\newcommand{\bra}[1]{\{#1\}}

\newcommand{\vi}{\vec{i}}

\def\bin{\mathrel{\bf in}}
\def\blet{\mathop{\bf let}}

\def\breturn{\mathop{\bf return}}

\def\bcase{\mathop{\bf case}}

\def\bfix{\mathop{\bf fix}}

\def\bof{\mathop{\bf of}}

\def\bsucc{{\bf s}}

\DeclareMathSymbol{:}{\mathrel}{operators}{`:}%
\DeclareMathSymbol{!}{\mathord}{operators}{`!}%

\def\sftrue{\mathsf{true}}

\newcommand{\algin}[4]{ {#1}; {#2}; {#3}; {#4}}
\newcommand{\algout}[4]{ #3; #4}
\def\produces{\Longrightarrow}

\def\emptyctx{\emptyset}

\newcommand{\qand}{\quad \text{and} \quad}

\usepackage{mathpartir}

\newtheorem{theorem}{Theorem}
\newtheorem{definition}[theorem]{Definition}
\newtheorem{lemma}[theorem]{Lemma}

\newtheorem{remark}[theorem]{Remark}

\ifdraft
\usepackage[draft,nomargin,inline]{fixme}
\else
\usepackage[final]{fixme}
\fi

\FXRegisterAuthor{bp}{anbp}{\color{DarkGreen}BP}
\FXRegisterAuthor{mg}{anmg}{\color{magenta}MG}
\FXRegisterAuthor{jh}{anjh}{\color{red}JH}
\FXRegisterAuthor{eg}{aneg}{\color{orange}EG}
\FXRegisterAuthor{aa}{anaa}{\color{teal}AA}

\def\eg{\egnote}
\def\jh{\jhnote}

\def\aa{\aanote}

\begin{document}

\iffull
\else

\setlength{\pdfpageheight}{\paperheight}
\setlength{\pdfpagewidth}{\paperwidth}

\conferenceinfo{IFL'14}{October 1st--3rd, 2014, Boston, MA, USA}
\copyrightyear{20yy}
\copyrightdata{978-1-4503-3284-2/14/10}
\doi{nnnnnnn.nnnnnnn}




\titlebanner{PREPRINT}        
\preprintfooter{PREPRINT}   
\fi

\title{Really Natural Linear Indexed Type Checking}

\iffull
\author{Arthur Azevedo de Amorim \qquad
  Emilio Jes\'us Gallego Arias \\
  Marco Gaboardi \qquad
  Justin Hsu}
\else
\authorinfo{Arthur Azevedo de Amorim}
           {University of Pennsylvania}{}
\authorinfo{Marco Gaboardi}
           {University of Dundee}{}
\authorinfo{Emilio Jes\'us Gallego Arias}
           {University of Pennsylvania}{}
\authorinfo{Justin Hsu}
           {University of Pennsylvania}{}
\fi

\maketitle

\begin{abstract}
  Recent works have shown the power of {\em linear indexed type systems} for
  enforcing complex program properties. These systems combine {\em linear types}
  with a language of {\em type-level indices}, allowing more fine-grained
  analyses. Such systems have been fruitfully applied in diverse domains,
  including implicit complexity and differential privacy.

  A natural way to enhance the expressiveness of this approach is by allowing
  the indices to depend on runtime information, in the spirit of dependent
  types.  This approach is used in \DFuzz, a language for differential privacy.
  The \DFuzz~type system relies on an index language supporting real and natural
  number arithmetic over constants and variables. Moreover, \DFuzz~uses a
  subtyping mechanism to make types more flexible. By themselves, linearity,
  dependency, and subtyping each require delicate handling when performing type
  checking or type inference; their combination increases this challenge
  substantially, as the features can interact in non-trivial ways.

  In this paper, we study the type-checking problem for \DFuzz.  We
  show how we can reduce type checking for (a simple extension of)
  \DFuzz~to constraint solving over a first-order theory of naturals
  and real numbers which, although undecidable, can often be handled
  in practice by standard numeric solvers.
\end{abstract}
\iffull\else


\fi

\section{Introduction}

{\em Linear indexed type systems} have been used to ensure safety
properties of programs with respect to different kinds of resources; examples
include usage analysis~\citep{Wadler91,conf/sas/WrightB93},
implicit complexity~\citep{GSS92,conf/esop/LagoS10,DLG11lics},
sensitivity analysis~\citep{Fuzz,DFuzz}, automatic timing
analysis~\citep{conf/popl/GhicaS11,conf/esop/GhicaS14}, and more.
Linear indexed types use a type-level \emph{index language}
to describe resources and \emph{linear types} to reason about the
program's resource usage in a compositional way.

One limitation of such systems is that resource usage is inferred independently
of the control flow of a program---e.g. the typing rule for branching usually
approximates resources by taking the maximal usage of one of the branches, and
recursion imposes even greater restrictions.
To improve this scenario, some authors have proposed extending such systems with
dependent types, using type indices to capture both resource usage and the
\emph{size information} of a program's input. This significantly enriches the
resulting analysis by allowing resource usage to depend on runtime information.
Linear dependently typed systems have been used in several domains, including
implicit complexity~\citep{DLG11lics,conf/ppdp/LagoP12} and sensitivity
analysis~\citep{DFuzz}.

Of course, there is a price to be paid for the increase in expressiveness: type
checking and type inference become inevitably more complex.
In linear indexed type systems, these tasks are often done in two stages: a
standard Hindley-Milner-like pass, followed by a constraint-solving procedure.
In some cases, the generated constraints can be solved automatically by using
custom algorithms~\citep{conf/aplas/LagoS10} or off-the-shelf SMT
solvers~\citep{FPCDSL,conf/esop/GhicaS14}. However, the constraints are specific
to the index language, and richer index languages often lead to more complex
constraints.

\subsection*{Type-checking \DFuzz}
In this paper we will focus on the type-checking problem for a particular
programming language with linear dependent types: \DFuzz~\citep{DFuzz}, a
dependently-typed extension of the \Fuzz programming language~\citep{Fuzz}.

\Fuzz uses linear indexed types to reason about programs in the context of
differential privacy. Its indices are real numbers that provide upper bounds on
the \emph{sensitivity} of a program, a quantity that measures the distance
between outputs on nearby inputs. In this setting, type checking and inference
result in a simple but effective static analysis for function sensitivity.
Indeed, as shown by~\citet{FPCDSL}, both of these can be performed efficiently
by using an SMT solver to discharge the numeric proof obligations arising from
the type system.

While \Fuzz works well on a variety of simple programs, it has a fundamental
limitation: sensitivity information cannot depend on runtime information, such
as the size of a data structure.
This is what \DFuzz is designed to handle.  \DFuzz indices combine information
about the \emph{size} of data structures with information about the {\em
  sensitivity} of functions. Technically, this is achieved by considering an
index language with variables ranging over integers (to refer to runtime sizes)
and reals (to refer to runtime sensitivities). This richer index language,
combined with dependent pattern-matching and subtyping, achieves increased
expressiveness in the analysis, providing sensitivity bounds beyond \Fuzz's
capabilities.

However, adding variables to the index language has a significant impact on the
difficulty of type checking. Concretely, since the index language also supports
addition and multiplication, index terms are now {\em polynomials} over the
index variables. Instead of constraints between real constants like in \Fuzz,
type checking constraints in \DFuzz may involve general polynomials.

A natural first approach is to try to extend the algorithm proposed
by~\citet{FPCDSL} to work with the new index language by simply
generating additional constraints when dealing with the new language
constructs. This would be similar in spirit to the work of
\citet{DLP13popl} for type inference for \dlPCF, a linear
dependently typed system for complexity analysis. A crucial difference
between that setting and \DFuzz is that the index language of
\dlPCF~can be extended by arbitrary (computable) functions. This makes
the approach to type inference for \dlPCF~proposed by Dal Lago and
Petit the most natural, since such functions can be used as direct
solutions to some of the introduced constraints.

However, such an approach does not work as well for \DFuzz, which opts for a
much smaller index language. While it may be possible to extend \DFuzz's index
language with general functions, we opt to keep the index language simple.
Instead, since the type system of \DFuzz also supports subtyping, we consider a
different approach inspired by techniques from the literature on subtyping~\cite{DBLP:conf/procomet/SteffenP94} and
 on constraint-based type-inference
approaches~\cite{DBLP:journals/tapos/OderskySW99,Pottier-Remy/emlti,Heeren02generalizinghindley-milner}.%

The main idea is to type-check a program by inferring some set of
sensitivities for it, and then testing whether the resulting type is a subtype
of the desired type. To obtain completeness (relative to checking the subtype), one must ensure that the
inferred sensitivities are the ``best'' possible for that
term. Unfortunately, the \DFuzz index language is not rich enough for
expressing such sensitivities. For instance, some cases require taking
the maximum of two sensitivity expressions, something that cannot be
done in the language of polynomials. We solve this problem by
extending the index language with three syntactic constructs,
resulting in a new type system that we name \EDFuzz. This new system
has meta-theoretic properties that are similar to those of \DFuzz, but also
simplifies the search for minimal sensitivities. Using
these new constructs, we design a sensitivity-inference algorithm for
\EDFuzz~which we show sound and complete, modulo constraint
resolution.

We now face the problem of solving the constraints generated by our
algorithm. First, we show how to compile the constraints generated by
the algorithmic systems to constraints in the first-order theory over
mixed integers and reals. This way, we can still use a numeric solver
without resorting to custom symbolic resolution.  Unfortunately, the
presence of natural numbers in the constraints has important
consequences: we show that \DFuzz type-checking is undecidable by
reducing from Hilbert's tenth problem, a standard undecidable problem.

While this result shows that we can't have a terminating type-checker
that is both sound and complete, not everything is lost. We first show that
by approximating the constraints, we obtain a sound and computable method to
type-check \EDFuzz programs. We show that this procedure can successfully
type-check a fragment of \EDFuzz which we call \UDFuzz; almost all of the
examples proposed by \citet{DFuzz} belong to this class. Of course, \UDFuzz is a
strict subset of \EDFuzz, and it is not hard to come up with well-typed programs
in \EDFuzz that are invalid under \UDFuzz.

Finally, we present a constraint simplification procedure that can significantly
reduce the complexity of our translated constraints (measured by the number of
alternating quantifiers), even when checking full \EDFuzz.

\subsection*{Contributions}

We briefly overview the \DFuzz programming language in
\Cref{sec:dfuzz}, to move to an informal exposition of the main
challenges involved in \Cref{sec:challenge}. Then, we present the
main contributions of the paper:
\begin{itemize}
\item \EDFuzz: an extension of \DFuzz with a more expressive sensitivity
  language that gives programs more precise types (\Cref{sec:edfuzz});
\item a sound and complete algorithm that reduces type checking and
  sensitivity inference in \EDFuzz to constraint solving over the first-order
  theory of $\N$ and $\R$ (\Cref{sec:algorithmic} and
  \Cref{sec:compilation});
\item a proof of undecidability of type checking in \DFuzz (and
  \EDFuzz)
  (\Cref{sec:undecidability});
\item a sound translation from the previous type-checking constraints
  to the first-order theory of the real numbers, a decidable theory
  (\Cref{sec:compilationReals}); and
\item a simplification procedure to make the constraints more amenable to
  automatic solving (\Cref{sec:simplification}).
\end{itemize}

Additionally, we have developed a prototype implementation of the
above, which we discuss in~\Cref{sec:implementation}.




\section{The \DFuzz System}
\label{sec:dfuzz}

\DFuzz~\citep{DFuzz} is a language for writing and verifying
differentially private programs. At its core lies a type system for
tracking function \emph{sensitivity}:
\begin{definition}
  Given two metric spaces $X,Y$, the sensitivity (or Lipschitz
  constant) of a function $f : X \to Y$ is a number $k$ such that
  $d_Y(f(x),f(x'))\leq k d_X(x,x')$ for all $x,x' \in X$. In this
  case, we say that $f$ is $k$-sensitive (or $k$-Lipschitz
  continuous).
\end{definition}
The precise relationship between differential privacy and function sensitivity
is beyond the scope of this paper; we refer the reader to previous
work~\citep{Fuzz,DFuzz} for more information. What is important for present
purposes is that \DFuzz uses a linear dependently-typed system for analyzing
function sensitivity. Let us begin with a brief presentation of \DFuzz before
discussing the type-checking challenges.

\subsection{Syntax and Types}
\DFuzz is an extension of PCF with dependent indexed linear types. Indices
consist of numeric constants; index-level variables, which range over {\em
  sizes} (natural numbers) or {\em sensitivities} (positive reals extended with
$\infty$, denoted $\Rextp$); and addition and multiplication of indices.  The
syntax of \DFuzz, including types, terms, and the index language, is shown in
\Cref{fig:dfuzz-terms}, which we briefly overview. Here, we omit some features
of the original system to keep our presentation simple.

\begin{itemize}
\item Abstraction and application for index variables are captured by the
  $\Lambda i : \kappa. e$ and $e[R]$ terms, with $\kappa$ representing the kind
  for $i$. We refer to variables of kind $\natone$ as {\em size variables},
  while variables of kind $\realone$ are {\em sensitivity variables}.
%
  \item Singleton types \eg{citation needed} $\N[S]$ and $\R[R]$ are used to
    related type-level sizes and sensitivities with term-level sizes and
    sensitivities.
  \item Dependent pattern matching over $\N[S]$ types is captured by the
    $\bcase$ construction.
  \item Linear functions indexed by $R$ are written $!_R \tyo \lin \tyw$.
  \item Variable environments $\ctxo$ carry an additional annotation for
    assignments $\cass{x}{R}{\tyo}$, representing the current sensitivity $R$
    for the variable $x$.
  \item Index variable environments $\tctxo$ specify the kinding of index
    variables.
  \item Constraint environments $\Phi$ store assumptions introduced under
    dependent pattern matching. Often, we will think of a constraint environment
    as the conjunction of its constraints.
\end{itemize}

\subsection{Environment Operations}
As in many similar systems, \DFuzz defines operations on variable
environments. Specifically, we can add two environments $\ctxo, \ctxw$, and
scale a single environment $\ctxo$ can by a sensitivity expression $R$.  We
define environment multiplication $R \cdot \ctxo$ as the operation taking every
element $\cass{x_i}{r_i}{\tyo_i}$ of $\ctxo$ to $\cass{x_i}{R\cdot
  r_i}{\tyo_i}$. Environment addition is defined iff all the common assignments
of $\ctxo$, $\ctxw$ map to the same type, that is to say, forall $x_i$ in
$\dom(\ctxo) \cap \dom(\ctxw)$, $(\cass{x_i}{R_i}{\tyo_i}) \in \ctxo \iff
(\cass{x_i}{S_i}{\tyo_i}) \in \ctxw$, where we write $\dom(\ctxo)$ for the
domain of an environment. In this case:
\begin{equation*}
  \begin{array}{lcll}
    \ctxo + \ctxw &=& \bra{\cass{x_i}{R_i + S_i}{\tyo} &\mid x_i \in
      \dom(\ctxo) \cap \dom(\ctxw)} \\
    &\cup& \bra{\cass{x_j}{R_j}{\tyo_j} &\mid x_j \in \dom(\ctxo) -
      \dom(\ctxw)} \\
    &\cup& \bra{\cass{x_k}{R_k}{\tyo_k} &\mid x_k \in
      \dom(\ctxw) - \dom(\ctxo)}
  \end{array}
\end{equation*}
\begin{figure}
  \centering
  \begin{equation*}
    \begin{array}[t]{lrl@{\hspace{-1cm}}r}
      \kappa & ::= & \realone \st \natone & \text{(kinds)}\\
      \Rextp & ::= & \R^{\ge 0} \cup \bra{\infty} & \text{(extended positive reals)}\\
      S & ::= & i \st 0 \st S + 1 & \text{(sizes)} \\
      R & ::= & \Rextp \st i \st S \st R + R \st R \cdot R
      &
      \text{(sensitivities)}\\
      \tyo,\tyw & ::= &
              \R
              \st \R[R]
              \st \N[S]
              \st !_{{R}} \tyo \multimap \tyw
              & \text{(types)}
              \\ & \st &
              \forall i : \kappa.\; \tyo
              \st \tyo \otimes \tyw
              \st \tyo\amp\tyw
              \\
        \tmo   & ::= & x
                \st \N
                \st \bsucc\ \tmo
                \st \R^{\geq 0}
                \st \bfix~(x : \tyo) .e
                & \text{(expressions)}
              \\ & \st &
                \lambda \cass{x}{R}{\tyo}. \tmo
                \st e_1\; e_2
              \\ & \st &
                \Lambda i : \kappa .\; \tmo
                \st \tmo[R]
              \\ & \st &
                \pair{e_1}{e_2}
                \st \pi_i\; e
              \\ & \st &
                (e_1, e_2)
                \st \blet~(x, y) = e \bin e'
              \\ & \st &
                \bcase e \bof 0 \Rightarrow e_0 \mid n_{[i]} + 1 \Rightarrow e_s
                \\
       \ctxo, \ctxw & ::= & \emptyset   \st \ctxo, \cass{x}{R}{\tyo} & \text{(environments)}\\
       \tctxo, \tctxw & ::= & \emptyset \st \tctxo, i : \kappa & \text{(sens. environments)}\\
       \cso, \csw & ::= & \top   \st \cso, S = 0 \st \cso, S = i + 1 & \text{(constraints)}\\
  \end{array}
  \end{equation*}
  \caption{\DFuzz Types and Expressions}
  \label{fig:dfuzz-terms}
\end{figure}

\subsection{Subtyping}
\DFuzz has a notion of subtyping, which intuitively corresponds to a
standard property of function sensitivity: a $k$-sensitive
function is also $k'$-sensitive for all $k' \ge k$. Furthermore, subtyping in
\DFuzz is the mechanism that allows types to use information from the constraint
environment; in this use, subtyping allows a form of type coercion.
We consider here a slightly
  simpler definition of subtyping than the one used in
  \citet{DFuzz}. In the environments we require subtyping to preserve
  the internal type. This slight modification will allow us to simplify
  some rules of the type-checking algorithm.

The semantics of the subtying relation is defined by interpreting
sensitivity expressions as functions that produce sensitivity
values. Formally, let $R$ be a sensitivity expression, well-typed
under environment $\phi$, and $\rho$ a suitable variable
\emph{valuation} (i.e., a function that maps each variable $i :
\kappa$ in $\phi$ to an element of $\lb \kappa \rb$, with $\lb \natone
\rb = \N$ and $\lb \realone \rb = \Rextp$). We then define
$\semu{R}_\rho$ as follows:
\begin{equation*}
  \begin{array}{rll@{\quad}l}
    \lb 0 \rb_\rho &:=& 0 \\
    \lb S + 1 \rb_\rho &:=& \lb S \rb_\rho + 1 \\
    \lb i \rb_\rho &:=& \rho(i) & \text{$i$ a variable} \\
    \lb r \rb_\rho &:=& r       & \text{$r$ a constant} \\
    \lb R_1 + R_2 \rb_\rho &:=& \lb R_1 \rb_\rho + \lb R_2 \rb_\rho \\
    \lb R_1 \cdot R_2 \rb_\rho &:=& \lb R_1 \rb_\rho \cdot \lb R_2 \rb_\rho
  \end{array}
\end{equation*}

Then, the standard ordering $\ge$ on $\Rextp$ induces an ordering on index
terms, which we can then extend to a subtype relation $\tless$ on types and
environments; the rules can be found in \Cref{fig:dfuzz-sub}. Note that checking
happens under the current constraint environment $\Phi$, so subtyping may use
information recovered from a dependent match.

The leaves of the subtype derivation are assertions $\tctxo; \cso \models R_1
\ge R_2$. These are defined logically as
\begin{equation*}
  \forall \rho \in \mathsf{val}(\tctxo). \semu{\cso}_\rho \Rightarrow
  \semu{R_1}_\rho \ge \semu{R_2}_\rho,
\end{equation*}
where $\mathsf{val}(\tctxo)$ is the set of all valid valuations for environment
$\tctxo$, and $\semu{\cso}_\rho$ is the conjunction of the denotations of each
formula in $\cso$, defined the usual way.

%
\begin{figure}
  \centering
  \begin{mathpar}
  \inferrule
   { }
   { \tctxo;\cso \vsub \tyo \tless \tyo  }                    \rname{\text{$\tless$-Refl}} \and
  \inferrule
   { \tctxo;\cso \vsub \tyo' \tless \tyo \\ \tctxo;\cso \vsub\tyw \tless \tyw' }
   { \tctxo;\cso \vsub \tyo \amp \tyw \tless \tyo' \amp \tyw'  }                       \rname{\struleamp} \and
  \inferrule
   { \tctxo;\cso \vsub \tyo \tless \tyo' \\ \tctxo;\cso \vsub\tyw \tless \tyw' }
   { \tctxo;\cso \vsub \tyo \otimes \tyw \tless \tyo' \otimes \tyw'  }                     \rname{\struletens} \and
  \inferrule
   { \tctxo;\cso \models R \le R' \\\\
     \tctxo ; \cso \vsub \tyo' \tless \tyo \\ \tctxo;\cso\vsub\tyw \tless \tyw' }
   { \tctxo;\cso \vsub !_R \tyo \multimap \tyw \tless !_{R'} \tyo' \multimap \tyw'  }                    \rname{\strulelin} \and
  \inferrule
   { \tctxo, i:\kappa; \cso \vsub \tyo \tless \tyw \\ \text{$i$ fresh in $\tctxo$} }
   { \tctxo;\cso \vsub \forall i: \kappa.\; \tyo \tless \forall i: \kappa.\; \tyw  } \rname{\struleforall} \and
  \inferrule
   { \forall(\cass{x}{R'}{\tyo})\in\ctxw, \exists R, (\cass{x}{R}{\tyo})\in\ctxo \wedge
     \left(\tctxo; \cso \models R_i \ge R'_i\right) }
   { \tctxo;\cso \vsub \ctxo  \tless \ctxw  } \rname{\text{$\tless$-Env}} \and
  \end{mathpar}
  \caption{\DFuzz Subtyping Relation}
  \label{fig:dfuzz-sub}
\end{figure}
\subsection{Typing}
\label{sec:typing}
Typing judgments for \DFuzz are of the form
\begin{equation*}
  \tctxo;\cso \mid \ctxo \vd \tmo : \tyo
\end{equation*}
meaning that term $\tmo$ has type $\tyo$ under environments $\tctxo$
and $\ctxo$ and constraints $\cso$; full rules are shown in
\Cref{fig:dfuzz-typing}.



We highlight here just the most complex rule, the dependent pattern-matching
rule $\rulenate$, which allows each branch to be typed under different
assumptions on the type $\N[S]$ of the scrutinee ($e$).  The left branch $e_0$
is typed under the assumption $S = 0$, while the right branch $e_s$ is typed
under the assumption $S = i + 1$ for some $i$. Combined with the rule for
fixpoints \rulefix, this allows us to express programs whose sensitivity depends
on the number of iterations or number of input elements. These rules also
require implicitly that all sensitivity (and size) expressions be well-typed
under the appropriate environments, which we note $\tctxo \vdash R$. Readers
interested in more details can consult \citet{DFuzz}; we follow their
presentation closely except for a few points, which we detail in the Appendix.
\newcommand{\dfuzzsubl}{
  \inferrule
  { \tctxo;\cso\mid \ctxw \vd \tmo : \tyo \\ \tctxo;\cso \vsub \ctxo \tless \ctxw}
  { \tctxo;\cso\mid \ctxo \vd \tmo : \tyo }                    \rname{\rulesubl}
}
\newcommand{\dfuzzsubr}{
  \inferrule
  { \tctxo;\cso\mid \ctxo \vd \tmo : \tyo \\ \tctxo;\cso \vsub \tyo \tless \tyw}
  { \tctxo;\cso\mid \ctxo \vd \tmo : \tyw }                    \rname{\rulesubr}
}
\newcommand{\dfuzzconst}{
  \inferrule
  { r \in \R }
  { \tctxo;\cso\mid \ctxo \vd  r : \R }                    \rname{\ruleconstR}
}
\newcommand{\dfuzznatc}{
  \inferrule
  { n = \lb S \rb }
  { \tctxo;\cso\mid \ctxo \vd  n : \N[S] }                 \rname{\ruleconstN}
}
\newcommand{\dfuzzvar}{
  \inferrule
        {  }
  { \tctxo;\cso\mid \ctxo, \cass{x}{1}{\tyo} \vd  x : \tyo }                \rname{\rulevar}
}
\newcommand{\dfuzzfix}{
  \inferrule
  { \tctxo;\cso\mid \ctxo, \cass{x}{\infty}{\tyo} \vd \tmo : \tyo  }
  { \tctxo;\cso \mid \infty\cdot \ctxo \vd \bfix~(x : \tyo).\tmo : \tyo }            \rname{\rulefix}
}
\newcommand{\dfuzzitens}{
  \inferrule
  { \tctxo;\cso\mid \ctxo_1 \vd e_1 : \tyo \\ \tctxo;\cso\mid \ctxo_2 \vd e_2 : \tyw }
  { \tctxo;\cso\mid \ctxo_1 + \ctxo_2 \vd (e_1, e_2) : \tyo \otimes   \tyw } \rname{\ruleitens}
}
\newcommand{\dfuzzetens}{
  \inferrule
  { \tctxo;\cso\mid \ctxw \vd e : \tyo \otimes \tyw  \\
    \tctxo;\cso\mid \ctxo, \cass{x}{R}{\tyo}, \cass{y}{R}{\tyw} \vd e' : \tyt }
  { \tctxo;\cso\mid \ctxo + R \cdot \ctxw
    \vd \blet~(x, y) = e \bin e' : \tyt}         \rname{\ruleetens}
}
\newcommand{\dfuzziamp}{
  \inferrule
  { \tctxo;\cso\mid \ctxo \vd e_1 : \tyo \\ \tctxo;\cso\mid \ctxo \vd e_2 : \tyw }
  { \tctxo;\cso\mid \ctxo \vd \pair{e_1}{e_2} : \tyo \amp   \tyw } \rname{\ruleiamp}
}
\newcommand{\dfuzzeamp}{
  \inferrule
  { \tctxo;\cso\mid \ctxo \vd e : \tyo_1 \amp \tyo_2 }
  { \tctxo;\cso\mid \ctxo \vd \pi_i \;e : \tyo_i }                  \rname{\ruleeamp}
}
\newcommand{\dfuzziapp}{
  \inferrule
    { \tctxo;\cso\mid \ctxo, \cass{x}{R}{\tyo} \vd  e : \tyw }
  { \tctxo;\cso\mid \ctxo \vd  \lambda \cass{x}{R}{\tyo}. e : !_R \tyo \multimap \tyw } \rname{\ruleiapp}
}
\newcommand{\dfuzzeapp}{
  \inferrule
  { \tctxo;\cso\mid \ctxo \vd  e_1 : !_R \tyo \multimap \tyw \\
    \tctxo;\cso\mid \ctxw \vd  e_2 : \tyo }
  { \tctxo;\cso\mid \ctxo + R\cdot\ctxw \vd e_1\;e_2 : \tyw }           \rname{\ruleeapp}
}
\newcommand{\dfuzzitapp}{
  \inferrule
    { \tctxo, i : \kappa;\cso \mid \ctxo \vd e : \tyo \\ \text{$i$ fresh in $\cso, \ctxo$}}
  {  \tctxo;\cso \mid \ctxo \vd  \Lambda i:\kappa .\; e : \forall i:\kappa.\; \tyo } \rname{\ruleitapp}
}
\newcommand{\dfuzzetapp}{
  \inferrule
    { \tctxo;\cso\mid \ctxo \vd e : \forall i : \kappa.\; \tyo
      \\ \tctxo \models S : \kappa}
  { \tctxo;\cso \mid \ctxo \vd e[S] : \tyo[S/i] }                   \rname{\ruleetapp}
}
\newcommand{\dfuzznats}{
  \inferrule
  { \tctxo;\cso\mid \ctxo \vd  e : \N[S] }
  { \tctxo;\cso\mid \ctxo \vd  \bsucc\ e : \N[S+1] }                \rname{\rulenats}
}
\newcommand{\dfuzznate}{
  \inferrule
  { \tctxo;\cso \mid \ctxw \vd e : \N[S] \\
    \tctxo ;\cso, S = 0     \mid  \ctxo \vd e_0 : \tyo \\\\
    \tctxo, i : \natone ;\cso, S = i + 1 \mid \ctxo, \cass{n}{R}{\N[i]} \vd e_s
  : \tyo \\
  i \text{ fresh in } \tctxo }
  { \tctxo;\cso\mid \ctxo + R \cdot \ctxw
    \vd \bcase e \breturn \tyo
    \bof 0 \Rightarrow e_0 \mid n_{[i]} + 1 \Rightarrow e_s : \tyo} \rname{\rulenate}
}
\begin{figure*}
  \centering
\begin{mathpar}
  \dfuzzsubl \and
  \dfuzzsubr  \\
  \dfuzzconst \and
  \dfuzznatc \\
  \dfuzzvar \and
  \dfuzzfix \\
  \dfuzziapp \and
  \dfuzzeapp \and
  \dfuzzitapp \and
  \dfuzzetapp \\
  \dfuzzitens \and
  \dfuzzetens \and
  \dfuzziamp \and
  \dfuzzeamp \\
  \dfuzznats \and
  \dfuzznate
\end{mathpar}
  \caption{\DFuzz Typing Rules}
  \label{fig:dfuzz-typing}
\end{figure*}

\subsection{Examples}
We close the overview of \DFuzz with some examples, to give an idea of
the increase in expressiveness brought by dependent types. We take the
liberty of including some features that were not introduced before to
make the examples more interesting.

We begin by considering multiplication of a real number by a natural number.
Without dependent types, the best type we can assign to multiplication is
$!_\infty \N \lin!_\infty\R\lin\R$, which is not very informative. However,
thanks to dependent types we can introduce a scaling primitive with the
following type:
\begin{equation*}
  \times : \forall i : \natone.\ !_\infty \N[i] \lin !_i \R \lin \R
\end{equation*}
By partially applying this operator, we obtain a scaling function with
the appropriate sensitivity, e.g.
\begin{equation*}
  (3 \times -) : !_{3} \R \lin \R .
\end{equation*}

\DFuzz uses probability distributions for differential privacy. The
type system includes a primitive for adding noise drawn from the
Laplace distribution to its input, with the following type:
\begin{equation*}
  \mathsf{add\_noise} : \forall \epsilon : \realone. !_\epsilon \R
  \lin \Circle \R
\end{equation*}
where $\Circle \R$ is the type of probability distributions over
$\R$. Here, $\epsilon$ is a parameter for controlling the amount of
added noise. This noise determines how ``far apart'' the resulting
distributions will be; as it turns out, given the distance function used
for probability distributions in \DFuzz, this results in an
$\epsilon$-sensitive function.

Finally (and more interestingly), the standard $\mathsf{map}$ function on lists
is given the following type in \DFuzz:
\begin{equation*}
  \mathsf{map} : \forall i\,R\,\tyo\,\tyw.
  !_i(!_R\tyo\lin\tyw)\lin!_{R}\mathsf{list}(\tyo)[i]\lin\mathsf{list}(\tyw)[i]
\end{equation*}
Here, $\mathsf{list}(\tyo)[i]$ is the type of lists of elements of some type
$\tyo$ with length equal to $i$. Because we have length-indexed lists, we can
correctly track the sensitivity of $\mathsf{map}$ on its function argument,
which is precisely the length of its list argument. \Fuzz, in contrast, would
require us to replace $i$ by $\infty$.

\section{The Challenge of Type-checking Linear Dependent Types}
\label{sec:challenge}
Type-checking a language with linear indexed types presents several
challenges, which are only compounded when dependent types and subtyping are
added to the mix.  In this section, we take a closer look at these
challenges.

%
\subsection{To Split, or not to Split?}
The first problem we face is due to linearity. Given a term and an environment,
we need a way to ``split'' the environment into appropriate subenvironments that
can be used in the recursive calls to type-check subterms.

Automatically inferring the right environments in our setting is
difficult, due to the index
language for \DFuzz.  Indeed, index terms are polynomials over index
variables, which may range over the reals or the naturals.  For
instance, we may know that a particular variable $x$ has sensitivity
$i^2 \cdot j^2 + 3$ in our environment.  However, it is not clear how
to split such sensitivity information between two environments that
share the variable $x$. In fact, as we will show below, in general it
is not always possible to find a split.
One might hope to simplify the type-checking task by requiring the programmer to
provide a few type annotations, like in non-linear type systems. Unfortunately,
this approach is impractical for the splitting problem because the annotations
must describe the split for every variable binding in the environment!

To better understand this obstacle, let us consider two general
approaches to type-checking linear type systems, which we call
the \emph{top-down} and \emph{bottom-up} strategies.

\subsubsection*{The Downfall of Top-Down}
For the type-checking
problem, suppose we are given the environment $\ctxo$, a term $e$, and a purported
type $\tyo$. The goal is to decide if $\ctxo \vd e : \tyo$ is derivable. The
{\em top-down} strategy takes an environment and a term, and attempts to partition
the environment and recursively type the subterms of $e$.

The main difficulty of this approach centers around splitting the
environment, a problem that is most clear in the application
rule. Here is a simplified version:
\begin{equation*}
  \inferrule
  {\ctxo \vd f : !_R \tyo \lin \tyw \\ \ctxw \vd e : \tyo }
        {\ctxo + R \cdot \ctxw \vd f~e : \tyw}
\end{equation*}
So given a type-checking problem $\ctxt \vd f~e : \tyo'$ our first
difficulty is to pick $R$, $\ctxo$, and $\ctxw$ such that $\ctxt =
\ctxo + R \cdot \ctxw$. We could try to guess $R$, but unfortunately
it may depend on the choice of $\ctxo$. Since our index language contains the
real numbers, the number of possible splittings isn't even finite.

A natural idea is to delay the choice of this split. For instance, we
may create a placeholder variable $R$ and placeholder environments
$\ctxo'$, $\ctxw'$, asserting $\ctxt = \ctxo' + R \cdot \ctxw'$ and
recursively type-checking $f$ and $e$. After reaching the leaves of
the derivation, we would have a set of constraints whose
satisfiability would imply that the program type-checks.

\aa{Find a way of putting this back: Unfortunately, this approach is
  largely impractical. First, it introduces a number of constraints so
  large that it was a performance problem for us in the past( This
  approach was tried with \Fuzz programs and indeed the quadratic
  factor made type-checking slower than acceptable) but that would be
  a more than acceptable price to pay if was not for that the
  generated constraints are}

Unfortunately, the constraints seem difficult to solve due to the
syntactical nature of our indices. In other words, the ``placeholder
variables'' are really {\em meta-variables} that range over index
terms, which could potentially depend on bound index variables. In
order to prove soundness of such a system with respect to the formal
typing system, the solver must return success {\em only} if there is
a solution where all the meta-variables can be instantiated to an
index term---a syntactic object. This is at odds with the way most
solvers work---\emph{semantically}---finding arbitrary solutions
over their domain. It is not clear how to solve these existential
constraints automatically for the specific index language of \DFuzz.

\subsubsection*{The Rise of Bottom-Up?}
A different approach is a {\em bottom-up} strategy: suppose we are
again given an environment $\ctxo$, a term $e$, and a type $\tyo$, and
we want to check if $\ctxo \vd e : \tyo$ is derivable. The main idea
is to avoid splitting environments by calculating the minimal
sensitivities needed for typing each subexpression. For each typing
rule, these minimal sensitivities can be combined
to find the resulting minimal sensitivities for $e$. Once this
is done, we just need to check whether these optimal sensitivities are
compatible with $\ctxo$ and $\tyo$ via subtyping.

Let's consider how this works in more detail by analyzing a few important cases.
At the base case, we type-check variables in a minimal environment (that is,
empty but for the variable) by assigning it the minimal sensitivity required:
\begin{equation*}
  \inferrule
           { }
  { \cass{x}{1}{\tyo} \vd x : \tyo }
\end{equation*}
Recall that we have weakening on the left so can add non-occurring
variables to the environment later.
%

Now, the key benefit of the bottom-up approach becomes evident in the
application rule: we can completely avoid the splitting problem. When
faced with a type-checking instance $\ctxt \vd f~e : \tyo$, we
recursively find optimal $\ctxo$, $R$, and $\ctxw$ for checking $f$
and $e$; then, checking that $\ctxt \tless \ctxo + R \cdot \ctxw$
suffices.

Unfortunately, things don't look so easy in the additive rules. Let's
examine the introduction rule for $\amp$:
\begin{equation*}
  \inferrule
  { \ctxo \vd e_1 : \tyo_1  \\  \ctxo \vd e_2 : \tyo_2 }
  { \ctxo \vd \pair{e_1}{e_2} : \tyo_1 \amp \tyo_2 }
\end{equation*}
This rule forces both environments to have the same sensitivities, but the bottom-up
idea may infer different environments for each expression:
\begin{equation*}
  \inferrule
  { \ctxo_1 \vd e_1 : \tyo_1  \\  \ctxo_2 \vd e_2 : \tyo_2 }
  { \ctxt? \vd \pair{e_1}{e_2} : \tyo_1 \amp \tyo_2 }
\end{equation*}
Now we need to guess a best environment $\ctxt?$, but the \DFuzz
sensitivity language is too weak to express this value.  For instance,
if we consider sensitivity expressions $r^2$ and $r$ depending on a
sensitivity variable $r$, we can show that there is no minimal
polynomial upper bound for them under the point-wise order on
polynomials\footnote{Indeed, it can be seen that \DFuzz does not
  possess minimal types. Refer to the Appendix for a more detailed
  proof.}.

To maintain the minimality invariant, we can extend the sensitivity
language with a new syntactic construct $\smax{R_1}{R_2}$ for
sensitivity-inference purposes only, which should denote the maximum of two
sensitivity values. We could then safely set $\ctxt?
:=\smax{\ctxo_1}{\ctxo_2}$, where the expression combines
sensitivities for the bindings on both environments as expected.

However, there is a problem with this approach:
the resulting algorithm is not sound with respect to the
original type system, because it allows more terms to be typed
even when sensitivities in the final type do not mention the new
construct! To see this, assume that our algorithm produces a
derivation $\ctxo' \vd e : \tyo'$ using extended sensitivities. Now,
soundness amounts to showing that for all $\ctxo$, $\tyo$ mentioning
only standard sensitivities such that $\ctxo \tless \ctxo'$ and $\tyo'
\tless \tyo$, there exists a typing derivation $\ctxo \vd e : \tyo$
that uses only the original sensitivity language. Let's try to sketch
how this proof would work by restricting our attention to a particular
instance of the application rule:
\begin{equation*}
  \inferrule
  {\phi; \emptyset \mid \emptyset \vd f : !_{R_f} \tyo \lin \tyw \\
    \phi; \emptyset \mid \cass{x}{\hat{R}_x}{\tyt} \vd e : \tyo }
        {\phi; \emptyset \mid \cass{x}{R_f\cdot \hat{R}_x}{\tyt} \vd f~e : \tyw}
\end{equation*}
where $\hat{R}_x$ is an extended sensitivity expression. By induction,
we know that for all standard sensitivity expressions $R_x$ such that
$R_x \ge \hat{R}_x$, we can obtain a standard derivation
$\cass{x}{R_x}{\tyt} \vd e : \tyo$. We also have standard $R_{xf}$
such that $R_{xf} \ge R_f \cdot \hat{R}_x$. Thus, all we need to do is
to calculate from $R_f$, $R_{xf}$ standard sensitivities $R'_{f}$,
$R'_x$ to be able to apply both induction hypotheses. The following
result shows that this is not always possible.
\begin{lemma}
  Given standard sensitivities expressions $R_{xf}$, $R_f$ and an
  extended sensitivity expression $\hat{R}_x$ such that $R_{xf} \ge
  R_f \cdot \hat{R}_x$, it is not the case that one can always find
  standard $R'_{f}$, $R'_x$ such that $R_{xf} \ge R'_{f} \cdot R'_x
  \land R'_{f} \ge R_f \land R'_x \ge \hat{R}_x$.
\end{lemma}
\begin{proof}
  Take $R_{xf}=r^2+1$, $R_f = r$ and $\hat{R}_x=\smax{2}{r}$. As we can see, we
  have $r^2+1 \ge r \cdot \smax{2}{r}$, with equality iff $r=1$. Suppose there
  exist standard sensitivity expressions $R'_{f}, R'_x$ that satisfy the
  statement. Because $R'_{f} \ge r$ and $R'_x \ge \smax{2}{r}$, we know by
  asymptotic analysis that the degree of $R'_{f}$ and $R'_x$ must be at least
  $1$. Furthermore, because $r^2+1 \ge R'_{f}\cdot R'_x$, their degree must be
  exactly $1$, with leading coefficient equal to $1$. Write $R'_{f} = r+a$ and
  $R'_x = r+b$, where $a,b$ are positive constants. The lower bound on $R'_x$
  implies $b \ge 2$. For $r = 1$, we have $R'_{f} \cdot R'_x \ge 3a+3 \ge
  3$. However, the lower and upper bounds for $R'_f \cdot R'_x$ coincide at that
  point, forcing $R'_{f} \cdot R'_x = 2$; contradiction. Thus, no such $R'_{f},
  R'_x$ can exist.
\end{proof}

It is not hard to adapt the above into a counterexample for the
soundness of the algorithm with respect to the original
system. However, we can recover soundness by extending the sensitivity
language for the basic typing rules as well.

\subsection{Avoiding the Avoidance Problem}
After the addition of least upper bounds for sensitivities, the
bottom-up approach is in a good working state for the basic
system. However, other constructs in the language introduce further
challenges. In particular, let's examine a simple version of the
abstraction rule for sensitivity variables:
\begin{equation*}
  \inferrule
  { \tctxo, i : \kappa \mid \ctxo \vd e : \tyo \\ \text{$i$ fresh in $\ctxo$} }
  { \tctxo \mid \ctxo \vd \Lambda i:\kappa .\; e : \forall i: \kappa.\; \tyo }
\end{equation*}
When this rule is interpreted in a top-down approach, usually no
problem arises; we would just introduce the new sensitivity variable
and proceed with type checking.

However, when the type-checking direction is reversed, we hit a version of the
avoidance problem
\citep{phd:lillibridge,Ghelli199875,Dreyer:2003:TSH:604131.604151}. The
avoidance problem usually appears in slightly different
scenarios related to existential types, and could be informally stated
as finding a best type free of a particular variable. In our case, we
must find the ``best'' $\ctxo$ free of $i$. It may not be obvious how $i$
could have been propagated to $\ctxo$, but indeed, a function $f$ in
$e$ could have a type such as $!_i \tyo\lin\tyw$, and applying $f$ will
introduce $i$ into the environment in the bottom-up approach.

Fortunately, in our setting, we can easily solve the avoidance problem
by further extending the sensitivity language. The ``best'' way of
freeing a sensitivity expression $R$ of a variable $i$ is to take the
supremum of $R$ over all possible values of $i$, which we denote by
$\ssup{i}{R}$\footnote{Contrary to $\smax{-}{-}$, it would have been
  possible to define this construct as a function over sensitivity
  expressions, without the need to extend their syntax. This would
  still be true even after introducing index-level case sensitivity
  expression for analyzing dependent pattern matching. As the
  translation is somewhat intricate and leads to more complex
  constraints, we chose to add it directly to the syntax of
  sensitivity expressions.}. Then, the minimal environment is
$\ssup{i}{\ctxo}$, where the supremum is extended to each binding in
the environment.
%

\subsection{Undependable Dependencies}
The last case to consider in our informal overview is
$\bcase$, also referred as dependent pattern matching.

The dependent pattern matching can be considered as a special case of
the two previous difficulties. Like the least upper bound, we must compute a least
upper bound of the resources used in two branches. However, now the information
coming from the successor branch may also contain sensitivities
depending on the newly introduced refinement variable, which cannot
occur in the upper bound; similar to the
avoidance problem we just discussed. On top of that, information coming from both sides
is conditional on the particular refinements induced by the match, so
any new sensitivity information that we propagate cannot really depend
on the refinements.

We now face a choice: we can introduce refinement types over
sensitivity and size variables of the form $\bra{\tyo \mid P(\vi)}$,
which would allow us to express the sensitivity inference for
$\bcase$ in term of the least upper bound and supremum
operations. However, we take a simpler path and add a conditional
operator on natural number expressions $S$, $\scase{S}{R_0}{i}{R_s}$,
interpreted as $R_0$ if $S$ is $0$ or $R_s[i \mapsto S - 1]$ if $S \ge
1$.

In the next sections we proceed to formally introduce the extended sensitivities
and its semantics; we discuss the type-checking algorithm, which depends on
solving inequality constraints over the extended sensitivities; and we study
several approaches to constraint solving and discuss decidability issues.

\section{Extended \DFuzz: \EDFuzz}
\label{sec:edfuzz}
\aa{What about just dropping $\hat{R}$ everywhere we can, and just using $R$
  instead?}  \jh{But they are different. $R$ is the old sensitivity (\DFuzz),
  while $\hat{R}$ are extended terms. Tried to make this clearer below.}%
\aa{I know they are, but we're not actually using them
  consistently. Cf. \Cref{sec:compilation}}
We define a conservative extension to \DFuzz's type system, \EDFuzz, which is
basically \DFuzz with an extended sensitivity language for the indices. We
summarize the new sensitivity terms, ranged over by meta-variable $\hat{R}$:
\begin{itemize}
\item $\smax{\hat{R}_1}{\hat{R}_2}$ is the pointwise least upper bound of
  sensitivity terms $\hat{R}_1, \hat{R}_2$.
\item $\ssup{i}{\hat{R}}$ is the pointwise least upper bound of $\hat{R}$ over
  all $i$.
\item $\scase{S}{\hat{R}_0}{i}{\hat{R}_s}$ is the conditional function on the
  size expression $S$ that is valued $\hat{R}_0$ when $S = 0$, and $\hat{R}_s[i
  \mapsto S - 1]$ when $S$ is a strictly positive integer.
\end{itemize}
The semantics of extended terms is defined as follows.

\begin{definition}[Extended sensitivity semantics]
  We extend the semantics of sensitivities to the new constructs in the
  following way (the old cases stay the same):
  \begin{align*}
    \lb\ssup{i:\kappa}{\esens}\rb_\rho
    &:= \sup_{r\in \lb\kappa\rb}\{\lb \esens\rb_{\rho\cup[i=r]}\}
    \\
    \lb\smax{\esens_1}{\esens_2}\rb_\rho
    &:= \max(\lb \esens_1 \rb_\rho, \lb \esens_2\rb_\rho)
    \\
    \lb\scase{S}{\esens_0}{i}{\esens_s}\rb_\rho
    &:= \left\{
      \begin{array}{lcl}
        \lb \esens_0\rb_{\rho} & \text{if} & \lb S\rb_{\rho} = 0 \\
        \lb \esens_s\rb_{\rho\cup [i = n - 1] }  & \text{if} & \lb S\rb_{\rho} = n \ge
        1.
      \end{array}
    \right .
  \end{align*}
\end{definition}
We define analogous operations on environments in the obvious way. For instance, if
$\cass{x}{R_1}{\tyo} \in \ctxo_1$ and $\cass{x}{R_2}{\tyo} \in \ctxo_2$, then
$\cass{x}{\smax{R_1}{R_2}}{\tyo} \in \smax{\ctxo_1}{\ctxo_2}$. As previously,
two-argument operations on environments are only defined when every variable that is
bound on both environments is assigned the same type by them.

It is not hard to show that any derivation valid in \DFuzz remains valid in
\EDFuzz.  Furthermore, \DFuzz's metatheory only relies on sensitivity terms
having an interpretation as total function from free variables to a
real number, rather than on any specific property about the
interpretation itself. The extended interpretation is total, and hence
the metatheory of \DFuzz extends to \EDFuzz.
\eg{Umm, we should improve this}

\section{Type Checking and Inference}
\label{sec:algorithmic}
We present a sound and complete type-checking and sensitivity-inference
algorithm for \EDFuzz. The algorithm assumes an oracle for deciding the
subtyping relation; in this sense, our algorithm is relatively complete. We
defer discussion about solving subtyping constraints to the next section.

The type-checking problem for \EDFuzz is the usual one: given a full context,
term, and type, the goal is to check if there is a derivation deriving the type
from the context.
\begin{definition}[Type Checking]
  Given an environment $\ctxo$, a term $e$, a type $\tyo$, the
  \emph{type-checking problem} for \EDFuzz is to determine whether a derivation
  $\emptyctx; \emptyctx \mid \ctxo \vd e : \tyo$ exists.
\end{definition}
Before we move to sensitivity inference, we introduce some notation for working
with contexts. It will be convenient to work with contexts with no top-level
annotations, i.e., contexts with bindings $(x : \tyo)$, where $\tyo$ is a proper
\EDFuzz type. We will call such contexts \emph{context skeletons}. For notation,
$\sclean{\ctxo}$ will mean the context $\ctxo$ with all top-level
annotations removed, while $\ctxskel$ will represent an arbitrary context
skeleton.

In our context, sensitivity inference means inferring the sensitivity
annotations in both an environment and a type. The input is an annotated
term\footnote{%
  We discuss annotations in \Cref{subsec:annot}}
and a context with without top-level annotations. The goal is to reconstruct a
type for the term, a full proper \EDFuzz context (e.g., with all top-level
annotations) along with a derivation, if possible.
\begin{definition}[Sensitivity Inference]
  Given an environment skeleton $\ctxskel$ and a term $e$, the
  \emph{sensitivity-inference problem} is to compute an environment $\ctxo$ and
  a type $\tyo$ with a derivation of $\emptyctx; \emptyctx \mid \ctxo \vd e :
  \tyo$, such that $\sclean{\ctxo} = \ctxskel$.
\end{definition}

\subsection{The Algorithm}

We can fulfill both goals using an algorithm that takes as inputs a
term $e$, an environment free of sensitivity annotations
$\cclean{\ctxo}$ and a refinement constraint $\cso$. The algorithm
will output an annotated environment $\ctxw$ and a type $\tyo$.
We write a call to the sensitivity inference algorithm as:
\begin{equation*}
  \algin{\tctxo}{\cso}{\ctxskel}{\tmo} \produces
  \algout{}{}{\ctxw}{\tyo}.
\end{equation*}
%
%
\newcommand{\dfuzzaalgsucc}{
  \inferrule
{ \algin{\tctxo}{\cso}{\ctxskel}{\tmo} \produces
  \algout{\emptyctx}{\sftrue}{\ctxo}{\N[\sizeitermone]} }
{ \algin{\tctxo}{\cso}{\ctxskel}{\bsucc\ \tmo} \produces
  \algout{\emptyctx}{\sftrue}{\ctxo}{\N[\sizeitermone+1]} }     \rname{\rulenats}
}

\newcommand{\dfuzzaalgconst}{
  \inferrule
        {  }
{ \algin{\tctxo}{\cso}{\ctxskel}{r} \produces
  \algout{\emptyctx}{\sftrue}{\mathrm{Ectx}(\ctxskel)}{\R} }   \rname{\ruleconst}
}

\newcommand{\dfuzzaalgnatc}{
  \inferrule
  { n = \lb S \rb }
{ \algin{\tctxo}{\cso}{\ctxskel}{n} \produces
  \algout{\emptyctx}{\sftrue}{\mathrm{Ectx}(\ctxskel)}{\N[S]} }  \rname{\ruleconstN}
}

\newcommand{\dfuzzaalgvar}{
  \inferrule
        { }
{ \algin{\tctxo}{\cso}{\ctxskel, \varone : \tyo}{\varone} \produces
  \algout{\emptyctx}{\sftrue}{\mathrm{Ectx}(\ctxskel), \cass{\varone}{1}{\tyo}}{\tyo} } \rname{\rulevar}
}

\newcommand{\dfuzzaalgiapp}{
  \inferrule
{ \algin{\tctxo}{\cso}{\ctxskel, \varone : \tyo}{\tmo} \produces
  \algout{\ectxo}{\cctxo}{\ctxo, \cass{\varone}{\sensitermone'}{\tyo}}{\tyw}
  \\\\
  \tctxo;\cso \models \sensitermone \geq \req{\sensitermone'} }
{ \algin{\tctxo}{\cso}{\ctxskel}{ \lambda (\cass{\varone}{\sensitermone}{\tyo}).~\tmo} \produces
  \algout{\ectxo}{\cctxo \wedge \sensitermone' \leq \sensitermone}
         {\ctxo}{!_\sensitermone \tyo \lin \tyw} }                          \rname{\ruleiapp}
}

\newcommand{\dfuzzaalgeapp}{
  \inferrule
  { \algin{\tctxo}{\cso}{\ctxskel}{\tmo_1}  \produces
    \algout{\ectxo_1}{\cctxo_1}{\ctxo}{!_\sensitermone \tyo \lin \tyw} \\\\
    \algin{\tctxo}{\cso}{\ctxwskel}{\tmo_2}  \produces
    \algout{\ectxo_2}{\cctxo_2}{\ctxw}{\tyo'}
    \\\\
  \tctxo;\cso \models \tyo' \tless \tyo }
    { \algin{\tctxo}{\cso}{\ctxskel}{\tmo_1\;\tmo_2} \produces
      \algout{\ectxo_1 , \ectxo_2}{\cctxo_1 \wedge \cctxo_2 \wedge \tyo' \tless \tyo}
      {\ctxo + \sensitermone \cdot \ctxw}{\tyw } }                           \rname{\ruleeapp}
}

\newcommand{\dfuzzaalgitapp}{
  \inferrule
  { \algin{\tctxo, \sizeivarone : \kappa}{\cso}{\ctxskel}{\tmo} \produces
    \algout{\ectxo}{\cctxo}{\ctxo}{\tyo}}
    { \algin{\tctxo}{\cso}{\ctxskel}{\Lambda \sizeivarone : \kappa.\; \tmo} \produces
      \algout{\ectxo_F}{\forall \sizeivarone : \kappa.\; \exists \ectxo. \cctxo \wedge \ctxo_F \tless \ctxo}
             {\ssup{i}{\ctxo}}{\forall \sizeivarone : \kappa.\; \tyo} }                        \rname{\ruleitapp}
}

\newcommand{\dfuzzaalgetapp}{
  \inferrule
  { \algin{\tctxo}{\cso}{\ctxskel}{\tmo} \produces
    \algout{\ectxo}{\cctxo}{\ctxo}{\forall \sizeivarone : \kappa.\; \tyo} \\
    \tctxo \models \sizeitermone : \kappa }
{ \algin{\tctxo}{\cso}{\ctxskel}{\tmo [\sizeitermone] } \produces
\algout{\ectxo}{\cctxo}{\ctxo}{\tyo[\sizeitermone / \sizeivarone] } }        \rname{\ruleetapp}
}

\newcommand{\dfuzzaalgiamp}{
  \inferrule
  { \algin{\tctxo}{\cso}{\ctxskel}{\tmo_1} \produces
    \algout{\ectxo_1}{\cctxo_1}{\ctxo_1}{\tyo_1} \\\\
    \algin{\tctxo}{\cso}{\ctxskel}{\tmo_2} \produces
    \algout{\ectxo_2}{\cctxo_2}{\ctxo_2}{\tyo_2}
  }
  { \algin{\tctxo}{\cso}{\ctxskel}{ \pair{\tmo_1}{\tmo_2} } \produces
  \algout{\ectxo_1, \ectxo_2, \ectxo_m}{\cctxo_1 \wedge \cctxo_2 \wedge \ctxo_m \tless \ctxo_1 \wedge \ctxo_m
  \tless \ctxo_2 }{\smax{\ctxo_1}{\ctxo_2}}{\tyo_1 \amp \tyo_2} }  \rname{\ruleiamp}
}

\newcommand{\dfuzzaalgeamp}{
  \inferrule
  { \algin{\tctxo}{\cso}{\ctxskel}{\tmo} \produces
    \algout{\ectxo}{\cctxo}{\ctxo}{\tyo_1 \amp \tyo_2}
  }
  { \algin{\tctxo}{\cso}{\ctxskel}{ \pi_i \tmo } \produces
    \algout{}{}{\ctxo}{\tyo_i} }  \rname{\ruleeamp}
}
\newcommand{\dfuzzaalgfix}{
  \inferrule
  { \algin{\tctxo}{\cso}{\ctxskel, \varone : \tyo}{\tmo} \produces
    \algout{}{}{\ctxo, \cass{\varone}{R}{\tyo}}{\tyo'} \\\\
    \tctxo; \cso \models \tyo' \tless \tyo
  }
  { \algin{\tctxo}{\cso}{\ctxskel}{\bfix \varone : \tyo.\; \tmo : \tyo}
    \produces
    \algout{}{}{\infty \cdot \ctxo}{\tyo}} \rname{\rulefix}
}
\newcommand{\dfuzzaalgitens}{
  \inferrule
  { \algin{\tctxo}{\cso}{\ctxskel}{\tmo_1} \produces
    \algout{\ectxo_1}{\cctxo_1}{\ctxo_1}{\tyo_1} \\\\
    \algin{\tctxo}{\cso}{\ctxskel}{\tmo_2} \produces
    \algout{\ectxo_2}{\cctxo_2}{\ctxo_2}{\tyo_2}
  }
  { \algin{\tctxo}{\cso}{\ctxskel}{ \pair{\tmo_1}{\tmo_2} } \produces
  \algout{}{}{\ctxo_1 + \ctxo_2}{\tyo_1 \otimes \tyo_2} } \rname{\ruleitens}
}

\newcommand{\dfuzzaalgetens}{
  \inferrule
  { \algin{\tctxo}{\cso}{\ctxskel}{\tmo} \produces
    \algout{\ectxo}{\cctxo}{\ctxw}{\tyo \otimes \tyw} \\\\
    \algin{\tctxo}{\cso}{\ctxskel, x : \tyo, y : \tyw}{\tmo'} \produces
    \algout{\ectxo}{\cctxo}{\ctxo, \cass{x}{R_1}{\tyo}, \cass{y}{R_2}{\tyw}}{\tyt}
  }
  { \algin{\tctxo}{\cso}{\ctxskel}{ \blet (x, y) = \tmo \bin \tmo'} \produces
    \algout{}{}{\ctxo + \smax{\req{R_1}}{\req{R_2}} \cdot \ctxw}{\tyt} }  \rname{\ruleetens}
}

\newcommand{\dfuzzaalgcasenat}{
  \inferrule
  { \algin{\tctxo}{\cso}{\ctxskel}{\tmo} \produces
    \algout{}{}{\ctxw}{\N[\sizeitermone]} \\
    \algin{\tctxo}{\cso, \sizeitermone=0}{\ctxskel}{\tmo_0} \produces
    \algout{}{}{\ctxo_0}{\tyo_0} \\\\
    \algin{\tctxo,\sizeivarone : \natone}{\cso,\sizeitermone=\sizeivarone+1}{\ctxskel,
      \varone : \N[i]}{\tmo_s} \produces
    \algout{}{}{\ctxo_s, \cass{\varone}{R'}{\N[i]}}{\tyo_s} \\\\
    \tctxo; \cso, S = 0 \models \tyo_0 \tless \tyo \\
    \tctxo, i : \natone; \cso, S = i + 1 \models \tyo_s \tless \tyo
  }
  { \algin{\tctxo}{\cso}{\ctxskel}{ \bcase \tmo \breturn \tyo \bof 0 \mapsto \tmo_0
      \mid \varone_{[\sizeivarone]} +1 \mapsto \tmo_s  } \\\\
      \produces
  \algout{\ectxo_0, \ectxo_s, \ectxo_m}{\cctxo_0 \wedge \cctxo_s \wedge \ctxo_m \tless \ctxo_0 \wedge \ctxo_m
  \tless \ctxo_s }{\scase{S}{\ctxo_0}{i}{\ctxo_s} + \scase{S}{0}{i}{\req{R'}} \cdot
  \ctxw}{\tyo} }  \rname{\rulenate}
}
\begin{figure*}
  \centering
  \begin{mathpar}
    \mprset{flushleft}
    \dfuzzaalgconst \and
    \dfuzzaalgnatc \and
    \dfuzzaalgvar \and
    \dfuzzaalgiapp \and
    \dfuzzaalgeapp \and
    \dfuzzaalgfix \and
    \dfuzzaalgitapp \and
    \dfuzzaalgetapp \and
    \dfuzzaalgitens \and
    \dfuzzaalgetens \and
    \dfuzzaalgiamp \and
    \dfuzzaalgeamp \and
    \dfuzzaalgsucc \and
    \dfuzzaalgcasenat
  \end{mathpar}
  \caption{Algorithmic Rules for \EDFuzz}
  \label{fig:dfuzza-alg}
\end{figure*}
\Cref{fig:dfuzza-alg} presents the full algorithm in a judgmental
style. The algorithm is based on a syntax-directed version of DFuzz
that enjoys several nice properties; full technical details and notation
definitions can be found in the Appendix. Here, we just sketch how the
transformation works in the proofs of soundness and completeness.
\begin{theorem}[Algorithmic Soundness]
  Suppose $\algin{\tctxo}{\cso}{\ctxskel}{\tmo} \produces
  \algout{\ectxo}{\cctxo}{\ctxo}{\tyo}$. Then, there is a derivation
  of $\tctxo; \cso \mid \ctxo \vd \tmo : \tyo$.
\end{theorem}
\begin{proof}
  We define two intermediate systems: The first one internalizing
  certain properties of weakening and a second, syntax-directed. The algorithm
  is a direct transcription of the
  syntax-directed system and soundness can be proved by induction on
  the number of steps. We prove soundness of the syntax-directed
  system by induction on the syntax-directed derivation.
\end{proof}
\begin{theorem}[Algorithmic Completeness]
  If
  $\tctxo;\cso\mid \ctxo \vd \tmo : \tyo$
  is derivable, then $\algin{\tctxo}{\cso}{\sclean{\ctxo}}{\tmo} \produces
  \algout{\ectxo}{\cctxo}{\ctxo'}{\tyo'}$ and $\tctxo; \cso \models
  \ctxo \tless \ctxo' \land \tyo' \tless \tyo$.
\end{theorem}
\begin{proof}
  We show that a ``best'' syntax-directed derivation can be build from
  any standard derivation by induction on the original derivation plus
  monotonicity and commutativity properties of the subtype
  relation. Completeness for the algorithm follows.
\end{proof}

\subsection{Removing Sensitivity Annotations} \label{subsec:annot}
We briefly discuss the role annotations play in our algorithm. \DFuzz programs
have three different annotations: the type of the argument for lambda terms
(including the sensitivity), the return type for case, and the type for
fixpoints.

The sensitivity annotations ensure that inferred types are free of terms with
extended sensitivities. This is useful for some optimizations on subtype
checking (introduced later in the paper). However, the general encoding of
subtyping checks works with full extended types, thus the sensitivity
annotations can be safely omitted and the system will infer types containing
extended sensitivities.

Due to technical difficulties in inferring the minimal sensitivity in the
presence of higher-order functions, the argument type in functions ($\tyo$
in $\lambda (x : \tyo)$) must be annotated, and we require the type of fixpoints
to be annotated.

\section{Constraint Solving}
\label{sec:compilation}
The type-checking algorithm introduced in the previous section produces
inequality constraints over the extended sensitivity language. While these
extended sensitivity terms may appear complicated, we can translate them into
equivalent formulas over the first-order theory of arithmetic over $\R$ and
$\N$.
While we show in the next section that the formulas we generate are usually
undecidable, they can still be handled by standard solvers.  Moreover,
in~\Cref{sec:compilationReals} we will present a sound (although not complete)
computable procedure to check the constraints.

To define our translation, it suffices to convert formulas with extended
sensitivities into equivalent ones that use only standard sensitivities, for we
can replace quantification over $\Rextp$ by equivalent formulas that only
quantify over $\R$ and $\N$. For instance, a formula of the form $\forall i :
\Rextp. P$, where $P$ has only quantifiers over $\R$ or $\N$, can be translated
into $(\forall i : \R. i \geq 0 \Rightarrow P)\wedge P'$, where $P'$ is the
result of substituting $\infty$ for $i$ in $P$ and performing all possible
simplifications.

The idea behind our translation is simple: we use a first-order formula to
uniquely specify each extended sensitivity term. Specifically, we define a
predicate $T(R)$ for each extended sensitivity term $R$, such that $\lb T(R)(r)
\rb_\rho$ holds exactly when $r$ is equal to the interpretation of $R$ under the
valuation $\rho$. For instance, consider the translation for $R_1 + R_2$:
\[
  T(R_1 + R_2)(r) := \exists r_1\,r_2 : \Rextp, T(R_1)(r_1) \wedge T(R_2)(r_2)
  \wedge r = r_1 + r_2.
\]
For $\rho$ a valuation for $R_1, R_2$, we have $r_1 = \lb R_1
\rb_\rho$ and $r_2 = \lb R_2 \rb_\rho$. Then the only $r$ that satisfies this
predicate is
\[
  r = r_1 + r_2 = \lb R_1 \rb_\rho + \lb R_2 \rb_\rho = \lb R_1 + R_2 \rb_\rho,
\]
as desired.

For a more involved example, consider the translation of $\smax{R_1}{R_2}$:
\begin{align*}
  &T(\smax{R_1}{R_2})(r) \\
  &:= \exists r_1\,r_2 : \Rextp, T(R_1)(r_1) \wedge
       T(R_2)(r_2) \wedge \\
       &(r_1 \ge r_2 \wedge r = r_1 \vee
        r_2 \ge r_1 \wedge r = r_2).
\end{align*}
Again, for any valuation $\rho$ of $R_1, R_2$, we have $r_1 = \lb R_1 \rb_\rho$
and $r_2 = \lb R_2 \rb_\rho$. The final conjunction states that $r$ must be the
largest among $r_1$ and $r_2$, which is precisely the semantics we have given
$\lb \smax{R_1}{R_2} \rb_\rho$. The full translation is in
\Cref{fig:const-trans}.

\begin{figure*}
\begin{align*}
  \kappa & := \N \st \Rextp \\
  T(i)(r) & := i = r \\
  T(R_1 + R_2)(r)
  & := \exists r_1\,r_2 : \Rextp, T(R_1)(r_1) \wedge
     T(R_2)(r_2) \wedge r = r_1 + r_2 \\
  T(R_1 \cdot R_2)(r)
  & := \exists r_1\,r_2 : \Rextp, T(R_1)(r_1) \wedge
       T(R_2)(r_2) \wedge r = r_1 \cdot r_2 \\
  T(\smax{R_1}{R_2})(r)
  & := \exists r_1\,r_2 : \Rextp, T(R_1)(r_1) \wedge
       T(R_2)(r_2) \wedge
       (r_1 \ge r_2 \wedge r = r_1 \vee
        r_2 \ge r_1 \wedge r = r_2) \\
  T(\scase{S}{R_0}{i}{R_s})(r)
  & := \exists r_s : \N, T(S)(r_s) \land
          (r_s = 0 \land T(R_0)(r)
           \lor \exists i : \N, r_s = i+1 \land T(R_s)(r)) \\
  T(\ssup{i : \kappa}{R})(r)
  & := \sbound(i : \kappa,R,r) \wedge
      \forall r'.
      \sbound(i : \kappa,R,r') \Rightarrow
        r' \ge r \\
  \sbound(i : \kappa,R,r)
  & := \forall i : \kappa. \exists r' : \Rextp. T(R)(r') \wedge r' \le r
\end{align*}
\caption{Constraint Translation}
\label{fig:const-trans}
\end{figure*}

We formalize our intuitive explanation of the translation with the following
lemma.
\begin{lemma}
  \label{lemma:sens-translation-sound}
  For every sensitivity expression $R$ and $r \in \Rextp$,
  and for every valuation $\rho$ whose domain
  contains the free variables of $R$, $\lb T(R)(r) \rb_\rho
  \iff r = \lb R \rb_\rho$
\end{lemma}

\begin{proof}
  By induction on $R$. We have already considered the $R_1 + R_2$ and
  $\smax{R_1}{R_2}$ cases above.
\end{proof}

Using the translation of terms, we can translate sensitivity constraints
generated by our typing algorithm.  We map each constraint of the form
\[
\icontextone; \iconstraintone \models R_1 \ge R_2
\]
to
\[
\forall \icontextone, \iconstraintone \Rightarrow \exists r_1\,r_2 : \Rextp,
T(R_1)(r_1)\wedge T(R_2)(r_2) \wedge r_1 \geq r_2
\]
Thanks to~\Cref{lemma:sens-translation-sound}, this translation is equivalent to
the semantics of sensitivity constraints given in~\Cref{sec:dfuzz}.

\section{Undecidability of Type-checking}
\label{sec:undecidability}
As we have seen in the previous section, constraints over our extended
sensitivity language can be translated to simple first-order formulas. Taken by
itself, this is not entirely satisfactory, as the first-order theory of $\N$ is
already undecidable. A nice illustration of this is Hilbert's tenth problem,
which asks if a polynomial equation of the form $P(\vec{x}) = 0$ over several
variables has any solutions over the natural numbers. After several years of
investigation, this property was finally shown to be undecidable.

In this section, we will show that this result makes \DFuzz type-checking
undecidable. We begin with an auxiliary lemma.

\begin{lemma} \label{lem:undec-aux}
  Given polynomials $P$, $Q$ over $n$ variables with coefficients in
  $\N$, checking $\forall \vec{i} \in \N^n, P(\vec{i}) \geq Q(\vec{i})$ is
  undecidable.
\end{lemma}

\begin{proof}
  We will use a solution to our problem to solve Hilbert's tenth
  problem. Suppose we are given a polynomial $P$ with integer
  coefficients, and we want to decide whether $\exists \vec{i} \in
  \N^n, P(\vec{i})=0$. This is equivalent to deciding $\neg \forall
  \vec{i} \in \N^n, P(\vec{i})^2 \geq 1$. Write $P(\vec{i})^2 =
  P^+(\vec{i}) - P^-(\vec{i})$, where $P^+$ and $P^-$ have only
  positive coefficients. Then our condition is equivalent to $\neg
  \forall \vec{i} \in \N^n, P^+(\vec{i}) \geq P^-(\vec{i}) + 1$. Thus,
  we can solve Hilbert's tenth problem by using $P^+$ and $P^-+1$ as
  inputs to our problem, which shows that it is undecidable.
\end{proof}

We can then show the following
\begin{theorem}
  \DFuzz type checking is undecidable.
\end{theorem}

\begin{proof}
  Suppose we are given $P$ and $Q$ as previously. Consider the types $\typeone =
  \forall \vec{i}, !_0 \N^n[\vec{i}] \lol \scale_{Q(\vec{i})} \R \lol \R$ and
  $\typetwo = \forall \vec{i}, !_0 \N^n[\vec{i}] \lol \scale_{P(\vec{i})} \R
  \lol \R$. Then $\typeone \tless \typetwo$ is equivalent to $\forall \vec{i},
  P(\vec{i}) \geq Q(\vec{i})$.  On the other hand, using recursion and dependent
  pattern matching, it is possible to write a function that multiplies a real
  number by a polynomial $Q(\vec{v})$ with variables ranging over $\N$. Its
  minimal type will clearly be $\typeone$. Therefore, type-checking it against
  $\typetwo$ is equivalent to deciding $\typeone \tless \typetwo$; since $P$ and
  $Q$ are arbitrary, this is undecidable by \Cref{lem:undec-aux}.
\end{proof}

\section{Approaches to Constraint Solving}

\aa{We should probably mention here that we don't really have a through
  empirical evaluation of these optimizations...}

Given that type-checking \DFuzz (and hence also \EDFuzz) is undecidable, is
there anything more we can do besides feeding the constraints to a solver and
hoping for the best? In this section, we discuss two possible directions to
tackle these constraints. For both of these approaches, we require that {\em all
  annotations in the term be standard sensitivities}, rather than
extended. Then, we have the following lemma. (We defer the proof to the
Appendix.)
\begin{lemma}[Standard Annotations]
  \label{lem:annot}
  Assume annotations in a term $e$ range over standard
  sensitivities and $\algin{\tctxo}{\cso}{\ctxskel}{\tmo} \produces
  \algout{\ectxo}{\cctxo}{\ctxo}{\tyo}$. Then:
  \begin{itemize}
    \item $\tyo$ has no extended sensitivities; and
    \item all constraints required for the algorithm are of the form
      $\tctxo;\cso\models R \geq R'$ where $R$ is a standard
      sensitivity term.
  \end{itemize}
\end{lemma}

\subsection{Modifying the subtype relation}
\label{sec:compilationReals}

The first approach is to restrict \EDFuzz to a fragment that enjoys decidable
type checking, which we call \UDFuzz. The main difference between both languages
is the interpretation of subtyping constraints: in \UDFuzz, constraint variables
are interpreted uniformly, ranging over \emph{all} possible sensitivity values,
regardless of their kind. As noted in~\Cref{sec:compilation}, we can translate
such formulas into the first-order theory of real arithmetic; since this theory
is decidable, so is \UDFuzz type checking.

Of course, this only makes sense if we can show that \UDFuzz is \emph{sound}
with respect to \EDFuzz. As it turns out, it suffices to restrict \UDFuzz
annotations to standard sensitivities---as we'll see, this forces the subtyping
relation of \UDFuzz to be a subrelation of the one of \EDFuzz. This restriction
rules out some programs that are typeable under \EDFuzz, but is expressive
enough to cover interesting ones, including most of the original examples
\citep{DFuzz}.

Formally, besides the restriction on annotations, \UDFuzz is the system obtained
from \EDFuzz by replacing all constraints of the form $\tctxo;\cso\models R \geq
R'$ with \emph{uniform} constraints $\tctxo;\cso\umodels R \geq R'$, which have
the following interpretation:
\[
\forall \rho \in \mathsf{val}_U(\tctxo).\lb\cso\rb_\rho^U\Rightarrow \lb
R\rb_\rho^U\geq\lb R'\rb_\rho^U
\]
Here, $\mathsf{val}_U(\tctxo)$ is the set of all \emph{uniform} valuations, that
map variables in $\dom(\phi)$ to values in $\Rextp$. The denotation
$\lb\cdot\rb_\rho^U$ of formulas and sensitivity and size terms is the same as
before, except for two cases:
\begin{align*}
\lb\ssup{i:\kappa}{\esens}\rb_{\rho}^U
&:= \sup_{r\in \Rextp}\{\lb \esens\rb_{\rho\cup[i=r]}^U\}
\\
\lb\scase{S}{\esens_0}{i}{\esens_s}\rb_{\rho}^U
&:= \left\{
  \begin{array}{lcl}
    \lb \esens_1\rb_{\rho}^U & \text{if} & \lb S\rb_{\rho}^U = 0 \\
    0 & \text{if} & \lb S\rb_{\rho}^U \in (0, 1) \\
    \lb \esens_2\rb_{\rho\cup [i = r - 1] }^U  & \text{if} & \lb
    S\rb_{\rho}^U = r \ge 1.
  \end{array}
\right .
\end{align*}

We first show that this uniform semantics is an extension of the standard
semantics.
\begin{lemma} \label{lem:unif-ext}%
  Suppose $R$ is a standard sensitivity term, typed under environment
  $\phi$. Then, for any standard valuation $\rho \in \mathsf{val}(\phi)$, we
  have
  \[
    \lb R \rb_\rho^U =  \lb R \rb_\rho.
  \]
\end{lemma}
\begin{proof}
  Immediate from the definition of the interpretation.
\end{proof}

We are now ready to prove that the uniform interpretation of constraints is
sound with respect to the original interpretation.
\begin{theorem} \label{thm:uniform-sound}%
  Suppose $R, R'$ are well-typed in environment $\phi$, with $R$ standard.
  Suppose that
  \(
    \tctxo; \cso \umodels R \geq R'
  \)
  is valid. Then
  \(
    \tctxo; \cso \models R \geq R'
  \)
  is also valid.
\end{theorem}
\begin{proof}
  It is clear that for any standard valuation $\rho \in \mathsf{val}(\phi)$, we
  have $\lb R' \rb_\rho^U \geq \lb R' \rb_\rho$. Assuming this, the hypothesis
  of the theorem yields
  \(
    \lb R \rb_\rho^U \geq \lb R' \rb_\rho^U \geq \lb R' \rb_\rho
  \)
  for every standard valuation $\rho\in\mathsf{val}(\phi)$.  But $R$ is a
  standard sensitivity, so $\lb R \rb_\rho^U = \lb R \rb_\rho$ by
  \Cref{lem:unif-ext}, and we are done.
\end{proof}
Thanks to \Cref{lem:annot}, all \UDFuzz constraints are of this form, which
shows that the subtype relation of \UDFuzz is a subrelation of the subtype
relation in \EDFuzz.  By reasoning analgous to
\Cref{lemma:sens-translation-sound}, we can show that relaxing the first order
translation of constraints captures this uniform interpretation. More formally:

\begin{lemma} \label{lem:relax-correct}
  For every sensitivity term $R$, let $T^U(R)$ be a unary predicate defined
  exactly as in \Cref{fig:const-trans}, but replacing quantification over $\N$
  with quantificiation over $\Rextp$ and with the modified $\bcase$ translation:
  \begin{equation*}
    \begin{array}{l}
      T^U(\scase{S}{R_0}{i}{R_s})(r) :=\\
      \quad
      \begin{array}{rl}
        \exists r_s : \Rextp,& T^U(S)(r_s) \land (r_s = 0 \land T^U(R_0)(r)) \\
        \lor & (0 < r_s < 1 \land r = 0) \\
        \lor & (\exists i : \Rextp, i\geq 0\land r_s = i+1 \land T^U(R_s)(r))
      \end{array}
    \end{array}
  \end{equation*}
  Then, $r \in \Rextp$, and for every uniform valuation $\rho$ whose
  domain contains the free variables of $R$, $\lb T^U(R)(r) \rb_\rho^U \iff r =
  \lb R \rb_\rho^U$.
\end{lemma}

By this lemma, we can give a sound, complete and decidable type-checking
algorithm for \UDFuzz.
\begin{theorem} \label{lem:dec-sound}
  Suppose we use our algorithmic system, with the constraints
  \[
    \tctxo; \cso \umodels R_1 \geq R_2
  \]
  handled by translation to the first order formula
  \[
    \forall \icontextone, \iconstraintone \Rightarrow \exists r : \Rextp,
    T^U(R_2)(r) \wedge R_1 \ge r,
  \]
  where all quantifiers are over $\Rextp$. Since the theory of $\Rextp$ is
  decidable, this gives an effective type-checking procedure for \UDFuzz.
\end{theorem}
\begin{proof}
  Note that $R_1$ is a standard sensitivity term, so the translated formula is
  indeed a first order formula over the theory of $\Rextp$. By
  \Cref{lem:relax-correct}, the translated formula is logically equivalent to
  \(
    \lb \iconstraintone \rb_\rho^U \Rightarrow \lb R_1 \rb_\rho^U \geq \lb R_2
    \rb_\rho^U
  \)
  for all uniform valuations $\rho\in\mathsf{val}_U(\phi)$, which in turn implies
  $\tctxo; \cso \models R_1 \geq R_2$ by \Cref{thm:uniform-sound}. This shows
  that the algorithmic system is sound and complete with respect to \UDFuzz.
\end{proof}

\begin{remark}
  \UDFuzz is a strict subset of \EDFuzz; informally, it contains \EDFuzz programs
  with typing derivations that do not use facts true over $\N$ but not over
  $\R$. One key way that subtyping is used in \EDFuzz is for equational
  manipulations of the indices; for instance, subtyping may be needed to change
  the index expression $3 (i + 1)$ to $3 i + 3$. This reasoning is available in
  \UDFuzz as well; indeed, most of the example programs in \DFuzz are
  typeable under \UDFuzz as well. (The only exception is $k$-medians, which
  extends the index language with a division function that we have not
  investigated.)

  However, there are many programs that lie in \EDFuzz but not in
  \UDFuzz---constraints as simple as $\forall i.\; i^2 \geq i$ are true when
  quantifing over the naturals but not when quantifying over the reals. Valid
  \EDFuzz programs that use these facts in their typing derivation will not lie
  in \UDFuzz.
\end{remark}

\subsection{Constraint Simplification} \label{sec:simplification}%
The second approach is to simplify the constraints generated by the translation
of \Cref{sec:compilation}, so that they can be better handled by solvers. Since
alternating quantifiers are a source of complexity in formulas, we devised a
rewriting procedure for producing constraints with no alternating
quantifiers. Here, we continue to require that all source annotations must be
standard sensitivity terms.

To begin, we generalize our three extended constructs with a new {\em
  constrained least upper bound} ($\bclub$) operation, with form
$\club{ \csens{\phi_1}{\Phi_1}{R_1}, \dots,
  \csens{\phi_n}{\Phi_n}{R_n}}$. Here, $\phi$ is a size and
sensitivity variable environment, $\Phi$ is a constraint environment, and $R$
is a sensitivity term, extended or standard. The judgment for a
well-formed $\bclub$ is
\[
  \phi \vdash \club{ \csens{\phi_1}{\Phi_1}{R_1}, \dots,
    \csens{\phi_n}{\Phi_n}{R_n}},
\]
where each $R_j$ has kind $\realone$ under $\phi, \phi_j; \Phi_j$, and $\phi,
\{\phi_j\}_j$ have disjoint domain. Intuitively, $\bclub$ is a maximum over a
set of sensitivities, restricting to sensitivities where the associated
constraint is satisfied. Sensitivities where the constraints are not satisfied
are ignored.  Formally, let $\phi$ contain the free variables of $\bclub$, and
let $\rho\in\mathsf{val}(\phi)$ be any standard valuation. We can give the
following interpretation of $\bclub$:
\begin{align*}
  &\llb \club{ \csens{\phi_1}{\Phi_1}{R_1}, \dots,
    \csens{\phi_n}{\Phi_n}{R_n}} \rrb_\rho := \\
  &\max_{j \in [n]} \max \{ \lb R_j
  \rb_{\rho \cup \rho_j} \mid \rho_j\in\mathsf{val}(\phi_j) \text{ and }
  \lb\Phi_j\rb_{\rho\cup\rho_j} \}.
\end{align*}
We define the  maximum over an empty set to be $0$.

Now, we can encode the extended sensitivity terms using only $\bclub$, through
the following translation function:
\begin{align*}
  C(\smax{\esens_1}{\esens_2})
  &:= \club{\csens{\emptyset}{\emptyset}{C(\esens_1)},
    \csens{\emptyset}{\emptyset}{C(\esens_2)}}
  \\
  C(\ssup{i}{\esens})
  &:= \club{\csens{i}{\emptyset}{C(\esens)}}
  \\
  C(\scase{S}{i}{\esens_0}{\esens_s})
  &:= \bclub\{\csens{\emptyset}{S = 0}{C(\esens_0)}, \\
  &\csens{i}{S = i + 1}{C(\esens_s)}\}
  \\
  C(\esens_1 + \esens_2) &:= C(\esens_1) + C(\esens_2)
  \\
  C(\esens_1 \cdot \esens_2) &:= C(\esens_1) \cdot C(\esens_2)
  \\
  C(R) &:= R \qquad \text{ otherwise.}
\end{align*}
While we may now have nested $\bclub$, we extend the interpretation in the
natural way.  We can show that the translation faithfully preserves the
semantics of the extended terms, with the following lemma.
\begin{lemma} \label{lem:club-correct}%
  Suppose $\phi \vdash R$ and $\rho\in\mathsf{val}(\phi)$ is a standard
  valuation.  Then, $\llb C(R) \rrb_\rho = \lb R \rb_\rho$.
\end{lemma}
\begin{proof}
  By induction on $R$.
\end{proof}
Now, we can simplify the compiled constraints. First, we can push all standard
sensitivity terms to the leaves of the expression. More formally, we have the
following lemma.
\begin{lemma} \label{lem:club-push}%
  Suppose $\phi \vdash R \cdot \club{ \csens{\phi_i}{\Phi_i}{C_i} }_i + R'$,
  where $R, R'$ are standard sensitivity terms, and $C_i$ is an arbitrary
  sensitivity term possibly involving $\bclub$. Then, for any standard valuation
  $\rho\in\mathsf{val}(\phi)$,
  \[
    \llb R \cdot \club{ \csens{\phi_i}{\Phi_i}{C_i} }_i + R' \rrb_\rho
    = \llb \club{ \csens{\phi_i}{\Phi_i}{R \cdot C_i + R'} }_i \rrb_\rho.
  \]
\end{lemma}
\begin{proof}
  By the definition of the interpretations, and the mathematical fact
  \(
    a \cdot \max_i \{ b_i \} + c = \max_i \{ a \cdot b_i + c \}
  \)
  for $a, b, c \geq 0$.
\end{proof}
Thus, without loss of generality we may reduce the compiled sensitivity
constraint to an expression of the form $Q$, with grammar
\[
  Q ::=
  \emptyset \st
  Q_1 + Q_2 \st
  Q_1 \cdot Q_2 \st
  \club{ \csens{\phi_i}{\Phi_i}{Q_i} } \st
  \club{ \csens{\phi_i}{\Phi_i}{R_i} },
\]
where $R_i$ are standard sensitivity terms. We will use the metavariable $V$ to
denote an arbitrary (possibly empty) collection of triples $(\phi_i; \Phi_i;
R_i)_i$, and the metavariable $W$ to denote an arbitrary (possibly empty)
collection of triples $(\phi_i; \Phi_i; Q_i)_i$.  Throughout, we will implicitly
work up to permutation of the arguments to $\bclub$: for instance, $\club{(X),
  (Y)}$ will be considered the same as $\club{(Y), (X)}$. We will also work up
to commutativity of addition and multiplication: $Q_1 + Q_2$ will be considered
the same as  $Q_2 + Q_1$, and likewise with multiplication.  We present the
constraint simplification rules as a rewrite relation $\cred$. As typical, we
will write $\cred^*$ for the reflexive, transitive closure of $\cred$. The full
rules are in \Cref{fig:constr-red}.

\begin{figure*}
  \centering
  \begin{mathpar}
  \inferrule
   { }
   { \club{ \csens{\phi}{\Phi}{\club{ \csens{\phi_i}{\Phi_i}{R_i}}_i}, V }
     \cred
     \club{ \csens{\phi \cup \phi_i}{\Phi \land
         \Phi_i}{R_i}, V} }_i                    \rname{\text{Flat}} \and
  \inferrule
   { }
   { \club{ \csens{\phi_i}{\Phi_i}{R_i} }_i + \club{ \csens{\phi_j'}{\Phi_j'}{R_j'} }_j
     \cred
     \club{ \csens{\phi_i \cup \phi_j'}{\Phi_i \land \Phi_j'}{R_i + R_j'} }_{ij} } \rname{\text{CPlus}} \and
  \inferrule
   { }
   { \club{ \csens{\phi_i}{\Phi_i}{R_i} }_i \cdot \club{ \csens{\phi_j'}{\Phi_j'}{R_j'} }_j
     \cred
     \club{ \csens{\phi_i \cup \phi_j'}{\Phi_i \land \Phi_j'}{R_i \cdot R_j'}
     }_{ij} }\rname{\text{CMult}} \\
  \inferrule
   { Q_1 \cred Q_1' }
   { Q_1 + Q_2 \cred Q_1' + Q_2' } \rname{\text{Plus}} \and
  \inferrule
   { Q_1 \cred Q_1' }
   { Q_1 \cdot Q_2 \cred Q_1' \cdot Q_2 } \rname{\text{Mult}} \and
  \inferrule
   { Q \cred Q' }
   { \club{ \csens{\phi}{\Phi}{Q}, W }
     \cred
     \club{ \csens{\phi}{\Phi}{Q'}, W } } \rname{\text{Red}} \and
  \end{mathpar}
  \caption{$\bclub$ Reduction}
  \label{fig:constr-red}
\end{figure*}
We can prove correctness of our constraint simplification with the following
lemma.
\begin{lemma} \label{lem:simpl-correct}%
  Suppose $Q \cred Q'$, and suppose $\phi \vdash Q$ and $\phi \vdash Q'$.  Then,
  for any standard valuation $\rho\in\mathsf{val}(\phi)$, we have $\llb Q
  \rrb_\rho = \llb Q' \rrb_\rho$.
\end{lemma}
\begin{proof}
  By induction on the derivation of $Q \cred Q'$. The cases $\text{Plus}$,
  $\text{Mult}$ and $\text{Red}$ are immediate by induction. The other cases all
  follow by the semantics of $\bclub$; details are in the Appendix.
\end{proof}
The simplification relation terminates in the following particular simple form.
\begin{lemma} \label{lem:simpl-normal}
  Let $Q$ be a sensitivity term involving $\bclub$. Along any reduction path,
  $Q$ reduces in finitely many steps to a term of the form
  \(
    \club{ V } = \club{ \csens{\phi_1}{\Phi_1}{R_1}, \dots,
      \csens{\phi_n}{\Phi_n}{R_n} } .
  \)
\end{lemma}
\begin{proof}
  First, note that any reduction of $Q$ must terminate in finitely many steps:
  by induction on the derivation of the reduction, it's clear that each
  reduction removes one $\bclub$ subterm, and no reductions introduce $\bclub$
  subterms. So, suppose that $Q$ is a term with no possible reductions.

  By induction on the structure of $Q$, we claim that $Q$ is of the desired
  form. Say if $Q = Q_1 + Q_2$, if either $Q_1, Q_2$ can reduce, then
  $\text{Plus}$ applies. If not, then by induction, $\text{CPlus}$ applies. The
  same reasoning follows for $Q = Q_1 \cdot Q_2$: either $\text{Mult}$ applies,
  or $\text{CMult}$ does. Finally, if $Q$ is a single $\bclub$ term, if
  $\text{Red}$ and $\text{Flat}$ both don't apply, then $Q$ is of the desired
  form.
\end{proof}
Finally, checking a constraint $\forall \phi.\, \Phi \Rightarrow R \geq
\club{V}$ is simple.
\begin{lemma} \label{lem:simpl-comp}
  Let $R$ be a standard sensitivity term, and let $V$ be
  \[
    V = \csens{\phi_1}{\Phi_1}{R_1}, \dots, \csens{\phi_n}{\Phi_n}{R_n}
  \]
  where each $R_j$ is a standard sensitivity term without $\bclub$. Then,
  $\tctxo; \cso \models R \geq \club{V}$ is logically equivalent to
  \[
    \forall \phi.\;
    \bigwedge_{j\in [n]}\forall \phi_j.\;\Phi \wedge \Phi_j \Rightarrow R \geq R_k.
  \]
\end{lemma}
\begin{proof}
  Immediate by the semantics of $\club{V}$.
\end{proof}
Putting together all the pieces, given a constraint
\(
  \tctxo; \cso \models R \geq R'
\), with $R$ standard, we can transform $C(R')$ to a term of the form $Q$ by pushing
all standard sensitivity terms to the leaves. Then, we normalize $Q \cred^*
\club{V}$ by \Cref{lem:simpl-normal} arbitrarily. By \Cref{lem:simpl-correct},
the interpretation of $Q$ and $\club{V}$ are the same, so we can reduce the
constraint $\tctxo;\cso \models R \geq \club{V}$ to a first order formula over
mixed naturals and $\Rextp$, with no alternating quantifiers, by
\Cref{lem:simpl-comp}.

\section{Implementation and Usability}
\label{sec:implementation}

We have implemented our algorithm for \EDFuzz, including the constraint
simplification described in the previous section, in a prototype type-checker.
The tool is written in OCaml, and uses the Why3 framework to check the generated
numeric inequalities with SMT solvers. We have successfully type-checked a range
of examples, including all but one of the examples from the original \DFuzz
paper. The remaining example involves a ``safe division'' operation on the
sensitivity language; we believe this operation can also be handled with our
techniques. The solvers had no problem solving the mixed natural/real
constraints on our examples, even though the problem is undecidable.

In our experience, the type-checker was quite usable. To give an idea of the
annotation burden in a typical example, consider the raw, annotated program
below.
\begin{lstlisting}
function cdf
 forall (i:size) (b:list(num)[i]) (db:[i]num bag)
 : list(num)[i] {
   listcase b of list(num)[i] {
     []          => nil @ [num]
   | x :: xs [m] =>
     let (lt, gt)  = bagsplit@[num]
       (fun (n:num) : bool {n < x}) db;
     let count  = (bagsize lt);
     let bigger = cdf[e][m] xs gt;
     cons @ [num][m] count bigger } }
\end{lstlisting}
This is a modified version of an original \DFuzz example. It uses a few
extensions to the system we have described, including additional primitive types
(\lstt{bag}) and lists with a basic form of polymorphism.

Our experience with error reporting was generally good. The tool points out
the location of the failed check, which was usually not far from the actual
error. The error messages leave a bit to be desired---usually, a polynomial
inequality that can't be proved---we leave improving this aspect to future work.

The implementation and examples are available
online.\footnote{\url{https://github.com/ejgallego/dfuzz}}

\section{Related work}
There is a vast literature on type checking for various combinations
of indexed types, linear types, dependent types and subtyping.
A distinctive feature of our approach is that our index language represents
natural and real number expressions. As we have shown in the previous
sections, this makes type checking non-trivial.

The work most closely related to ours is
\citet{DLP13popl}, who studied the type-inference problem for \dlPCF,
a relatively-complete type system for complexity analysis introduced
in \citet{DLG11lics}. \dlPCF~uses ideas similar to \DFuzz~but brings
the idea of linear dependent types to the limit. Indeed, \dlPCF~index
language contains function symbols that are given meaning by an
equational program.  The equational program then plays the role of an
oracle for the type system---\dlPCF~is in fact a family of type
systems parametrized over the equational program.  The main
contribution of \citet{DLP13popl} is an algorithm that, given a
\PCF~program, generates a type and the set of constraints that must
be satisfied in order to assign the return type to the input
term.

In our terminology, their work is similar to the top-down approach we
detailed in \Cref{sec:challenge}.  As we discussed there, the complication
of this approach is that it requires solving constraints over
expressions---with possible function symbols---of the index-level
language.  As shown by Dal Lago and Petit, a clear advantage of the
\dlPCF~formulation is that instead of introducing an existential
variable over expressions, one can introduce a new function symbol
that will then be given meaning by the equational program generated by
the constraints---i.e., the constraints give a description of the
semantics of the program, which can be turned in an
equational program, that in turn gives meaning to the function symbols
of the index language appearing in the type.  Clearly, this approach
cannot be reduced to numeric resolution and need instead a combination
of numeric and symbolic solving technology.  The authors show that these
constraints can be anyway handled by using the Why3 framework. Some
constraints are discharged automatically by some of the solvers
available in Why3 while others requires an interactive resolution
using Coq.

As explained in \Cref{sec:challenge}, the situation
with \DFuzz~is different. Indeed, \DFuzz can be seen as a simplified
version of \dlPCF---simplifying in particular the typing for the
fixpoint and without variable bindings in !-types---extended however
to deal with indices representing real numbers and using
quantifications over index variables. A key distinction of \DFuzz~is that the
set of constructors for the language of sensitivity is
{\em fixed}---one cannot add arbitrary functions. Moreover, the extension to real numbers gives a different
behavior from how natural numbers are used in \dlPCF---e.g., our
example for the lack of minimal type would make no sense in \dlPCF.
These distinctions make the type checking problem very different.

For another approach that is closely related to our work, recall that \DFuzz is
an extension of \Fuzz. The sensitivity-inference and sensitivity-checking
problems for \Fuzz have been studied in~\citet{FPCDSL}. These problems are
simpler than the one studied here since in \Fuzz~there is no dependency, no
quantification and no subtyping. Indeed, the constraints generated are much
simpler and can be solved quickly by an SMT solver.

Similarly, \citet{DBLP:conf/csfw/EignerM13} have studied an extension of \Fuzz
for modeling protocols. In their work they also give an algorithmic version of
their type system. Their type system presents challenges similar to \Fuzz, which
they handle with algebraic manipulations.  More precisely, their algorithmic
version uses a technique similar to the one developed in
\citet{Cervesato2000133} for the splitting of resources: when a rule with
multiple premises is encountered the algorithmic system, first allocate all the
resources to the first branch and then allocate the remaining resources to the
second branch. Unfortunately, this approach cannot be easily applied to
\DFuzz~due to the presence of index variables and dependent pattern matching.

From a different direction, recent works \citep{conf/esop/Brunel14, conf/esop/GhicaS14} have shown how
linear indexed type systems can be made more abstract and useful to
analyze abstract resources. In particular, this kind of analyses is
connected to comonadic notions of computations~\citep{petricek2013coeffects}. The type-inference
algorithm described in \citet{conf/esop/GhicaS14}  is parametric on an
abstract notion of resource. This resource can be instantiated on a language
for sensitivities similar to the one in \Fuzz. So, this abstract type-inference
procedure could be also used for sensitivity analysis.

\DFuzz is one of several languages combining linear and dependent types. For
example, ATS~\citep{conf/icfp/ChenX05} is designed around a dependent type
system enriched with a notion of resources that is a type-level representation
of memory locations; these resources are managed using a linear discipline.  ATS
uses these features to verify the correctness of memory and pointer management.

Even if the use of linear types in ATS is very different from the one presented
here, our type checking algorithm shares some similarities with ATS's
one. The main difference is that ATS uses interactive
theorem proving to discharge proof obligations while, thanks to the
restricted scope of our analysis, our constraints can be handled by
numeric solvers. In contrast, DML~\citep{DML99}---a predecessor of ATS which did not
use linear types---uses an approach similar to ours by solving proof
obligations using automatic numeric resolution. This required limitations on the
operations available in the index language, similar to \DFuzz.

Another work considering lightweight dependent types is the one
by~\citet{vmcai2013-zhu}. In particular they propose a technique based
on dependent types to reduce the verification of higher order
programs to the verification of a first order language. While the goal
of their work is similar in spirit to ours, their technique has only
superficial similarities with the one presented here.

Finally, our work has been informed by the wide literature on type-checking, far
too large to summarize here.  For instance, the problem of dealing with
subtyping rules by using syntax-directed systems has been studied by
\citet{DBLP:conf/procomet/SteffenP94}, and others.

\section{Conclusions and Future Work}
We have presented a type-checking and sensitivity-inference algorithm for
\EDFuzz---a simple extension of \DFuzz---featuring a linear indexed dependently
type system. While we have shown that \DFuzz type checking is undecidable in the
general case, our approach generates constraints over the first-order theory
over the reals and naturals, for which there are standard (though necessarily
incomplete) solvers.

Overall, our design was guided by two principles: to stay as close to
\DFuzz as possible, and to provide a practical type checking
procedure. While we do require extensions to \DFuzz, there is a clear
motivation for the introduction of each new construct. The idea of
making a limited enrichment of the index language in order to simplify
type-checking may be applicable to other linear indexed type
systems. Furthermore, designers of such systems would do well to keep
implementability in mind: seemingly unimportant decisions that
simplify the metatheory may have a serious impact on type-checking.







\bibliographystyle{abbrvnat}
\bibliography{header,fuzz}

\appendix

\newpage
\onecolumn
\section{Differences Compared to \citet{DFuzz}}
While we hew closely to the presentation of \DFuzz in \citet{DFuzz}, we make a
few technical changes.

\begin{itemize}
  \item The environment weakening operation $\ctxo \tless \ctxo'$ in \DFuzz allows
    the types to change. That is, a binding $\cass{x}{R}{\tyo} \in \ctxo$ can be
    weakened to $\cass{x}{R'}{\tyo'}$ for $\tyo \tless \tyo'$ two syntactically
    different types. We take a more restricted weakening rule, where the types
    must be syntactically the same; we are unaware of any programs that need the
    more general rule.
  \item We take the interpretation of $\infty \cdot 0$ to be $\infty$, rather
    than $0$.
  \item We assume some additional type annotations in the source language, as
    discussed in \Cref{sec:algorithmic}
\end{itemize}

\section{The \DFuzz$_\Box$ system}

The first system has the goal to enjoy environment ``uniformity'', in the
sense that sensitivity information in the environments may be missing. We
denote such an assignment $\cempty{x}{\tyo}$. This is a subtle
technical point for crucial to enable syntax-directed typability.

We modify subtyping for environments such that $\ctxo \tless \ctxw$
requires $\ctxo$, $\ctxw$ to have the same domain. The new rule is:
\begin{equation*}
  \inferrule
   { \forall (\cass{x_i}{R_i}{\tyo_i}, \cass{x_i}{R'_i}{\tyo_i}) \in (\ctxo, \ctxw) \\\\
     \dom(\ctxw) = \dom (\ctxo) \\
     \tctxo; \cso \models R_i \ge R'_i \lor R'_i = \Box}
   { \tctxo;\cso \vsub \ctxo  \tless \ctxw  } \rname{\text{$\tless$-Env}}
\end{equation*}
This subsumes regular variable weakening.
Environment operations must be aware of $\Box$, with $\Box + i = i$, $i
\cdot \Box = \Box$ for the annotations.

\begin{definition}[Box erasure]
  For any environment $\ctxo$, we define the $\Box$-erasure operation
  $|\ctxo| = \bra{\cass{x}{R}{\tyo} \mid \cass{x}{R}{\tyo} \in \ctxo
    \land R \neq \Box}$.
\end{definition}

We introduce the $\Box$ system in \Cref{fig:dfuzz-box-typing}.
\newcommand{\dfuzzbsubl}{
  \inferrule
  { \tctxo;\cso\mid \ctxw \vd \tmo : \tyo \\ \tctxo;\cso \vsub \ctxo \tless \ctxw}
  { \tctxo;\cso\mid \ctxo \vd \tmo : \tyo }                    \rname{\rulesubl}
}
\newcommand{\dfuzzbsubr}{
\inferrule
  { \tctxo;\cso\mid \ctxo \vd \tmo : \tyo \\ \tctxo;\cso \vsub \tyo \tless \tyw}
  { \tctxo;\cso\mid \ctxo \vd \tmo : \tyw }                    \rname{\rulesubr}
}
\newcommand{\dfuzzbconst}{
  \inferrule
  {  }
  { \tctxo;\cso\mid \ctxo \vd  \realone : \R }
  \rname{\ruleconstR}
}
\newcommand{\dfuzzbnatc}{
  \inferrule
  { n = \lb S \rb }
  { \tctxo;\cso\mid \ctxo \vd  n : \N[S] }                     \rname{\ruleconstN}
}
\newcommand{\dfuzzbvar}{
  \inferrule
        {  }
  { \tctxo;\cso\mid \ctxo, \cass{x}{1}{\tyo} \vd  x : \tyo }   \rname{\rulevar}
}
\newcommand{\dfuzzbitens}{
  \inferrule
  { \tctxo;\cso\mid \ctxo_1 \vd e_1 : \tyo \\ \tctxo;\cso\mid \ctxo_2 \vd e_2 : \tyw }
  { \tctxo;\cso\mid \ctxo_1 + \ctxo_2 \vd (e_1, e_2) : \tyo \otimes   \tyw } \rname{\ruleitens}
}
\newcommand{\dfuzzbetens}{
  \inferrule
  { \tctxo;\cso\mid \ctxw \vd e : \tyo \otimes \tyw  \\
    \tctxo;\cso\mid \ctxo, \cass{x}{R}{\tyo}, \cass{y}{R}{\tyw} \vd e' : \tyt \\ R \neq \Box}
  { \tctxo;\cso\mid \ctxo + R \cdot \ctxw \vd \blet (x, y) = e \bin e' : \tyt}         \rname{\ruleetens}
}
\newcommand{\dfuzzbetensold}{
  \inferrule
  { \tctxo;\cso\mid \ctxw \vd e : \tyo \otimes \tyw  \\
    \tctxo;\cso\mid \ctxo, \cass{x}{R}{\tyo_1}, \cass{y}{R}{\tyw} \vd e' : \tyt \\
    R \neq \Box}
  { \tctxo;\cso\mid \ctxo + R \cdot \ctxw
    \vd \blet (x : \tyo, y : \tyw) = e \bin e' : \tyt}         \rname{\ruleetens}
}
\newcommand{\dfuzzbiamp}{
  \inferrule
  { \tctxo;\cso\mid \ctxo \vd e_1 : \tyo \\ \tctxo;\cso\mid \ctxo \vd e_2 : \tyw }
  { \tctxo;\cso\mid \ctxo \vd \pair{e_1}{e_2} : \tyo \amp   \tyw } \rname{\ruleiamp}
}
\newcommand{\dfuzzbeamp}{
  \inferrule
  { \tctxo;\cso\mid \ctxo \vd e : \tyo_1 \amp \tyo_2 }
  { \tctxo;\cso\mid \ctxo \vd \pi_i \;e : \tyo_i }                  \rname{\ruleeamp}
}
\newcommand{\dfuzzbiapp}{
  \inferrule
    { \tctxo;\cso\mid \ctxo, \cass{x}{R}{\tyo} \vd  e : \tyw \\
      R \neq \Box
    }
  { \tctxo;\cso\mid \ctxo \vd  \lambda (\cass{x}{R}{\tyo}). e : !_R \tyo \multimap \tyw } \rname{\ruleiapp}
}
\newcommand{\dfuzzbeapp}{
  \inferrule
  { \tctxo;\cso\mid \ctxo \vd  e_1 : !_R \tyo \multimap \tyw \\
    \tctxo;\cso\mid \ctxw \vd  e_2 : \tyo }
  { \tctxo;\cso\mid \ctxo + R\cdot\ctxw \vd e_1\;e_2 : \tyw }           \rname{\ruleeapp}
}
\newcommand{\dfuzzbitapp}{
  \inferrule
    { \tctxo, i : \kappa;\cso \mid \ctxo \vd e : \tyo \\ \text{$i$ fresh in
    $\cso, \ctxo$}}
  {  \tctxo;\cso \mid \ctxo \vd  \Lambda i:\kappa .\; e : \forall i:\kappa.\; \tyo } \rname{\ruleitapp}
}
\newcommand{\dfuzzbetapp}{
  \inferrule
    { \tctxo;\cso\mid \ctxo \vd e : \forall i : \kappa.\; \tyo
      \\ \tctxo \models S : \kappa}
  { \tctxo;\cso \mid \ctxo \vd e[S] : \tyo[S/i] }                   \rname{\ruleetapp}
}
\newcommand{\dfuzzbfix}{
  \inferrule
  { \tctxo; \cso \mid \ctxo, \cass{\varone}{\infty}{\tyo} \vd e : \tyo }
  { \tctxo; \cso \mid \infty \cdot \ctxo \vd \bfix \varone : \tyo.\; e : \tyo } \rname{\rulefix}
}
\newcommand{\dfuzzbnats}{
  \inferrule
  { \tctxo;\cso\mid \ctxo \vd  e : \N[S] }
  { \tctxo;\cso\mid \ctxo \vd  e + 1 : \N[S+1] }                    \rname{\rulenats}
}
\newcommand{\dfuzzbnate}{
  \inferrule
  { \tctxo;\cso \mid \ctxw \vd e : \N[S] \\
    \tctxo ;\cso, S = 0     \mid  \ctxo \vd e_0 : \tyo \\\\
    \tctxo, i : \natone ;\cso, S = i + 1 \mid \ctxo, \cass{n}{R}{\N[i]} \vd e_s
  : \tyo \\ i \# R \\ R \neq \Box }
  { \tctxo;\cso\mid \ctxo + R \cdot \ctxw
    \vd \bcase e \breturn \tyo
    \bof 0 \Rightarrow e_0 \mid n_{[i]} + 1 \Rightarrow e_s : \tyo} \rname{\rulenate}
}

\begin{figure*}
  \centering
\begin{mathpar}
  \dfuzzbsubl \and
  \dfuzzbsubr \and
  \dfuzzbconst \and
  \dfuzzbnatc \and
  \dfuzzbvar \and
  \dfuzzbitens \and
  \dfuzzbetens \and
  \dfuzzbiamp \and
  \dfuzzbeamp \and
  \dfuzzbiapp \and
  \dfuzzbeapp \and
  \dfuzzbitapp \and
  \dfuzzbetapp \and
  \dfuzzbfix \and
  \dfuzzbnats \and
  \dfuzzbnate
\end{mathpar}
  \caption{\DFuzz$_\Box$ Type Judgment}
  \label{fig:dfuzz-box-typing}
\end{figure*}

We prove that derivations in a system with $\Box$ are in direct
correspondence with derivation in a system without it.

\begin{lemma}
  Assume $\tctxo; \cso \mid \ctxo \vd e : \tyo$ in the $\Box$ system,
  then $\tctxo; \cso \mid |\ctxo| \vd e : \tyo$ in the system without
  it.
\end{lemma}
\begin{proof}
  By induction on the typing derivation. The base cases and cases
  where the environment is not modified are immediate. Subtyping on the
  left is proven by weakening.

  The rest of cases are split in two:
  \begin{itemize}
  \item All cases featuring variables in the top rule, also have the
    condition $R \neq \Box$, this is enough.
  \item For the cases involving environment operations, the proofs is
    completed by following properties:
    \begin{equation*}
      | R \cdot \ctxo | = R \cdot |\ctxo|
      \qquad
      | \ctxo + \ctxw | = |\ctxo| + | \ctxw|
    \end{equation*}
  \end{itemize}

\end{proof}
\begin{lemma}
  Assume $\tctxo; \cso \mid \ctxo \vd e : \tyo$ in the system without
  $\Box$, then $\tctxo; \cso \mid \ctxo \vd e : \tyo$ in the system
  with it.
\end{lemma}
\begin{proof}
  The proof is mostly routine by induction on the derivation, but
  relies in the following fact of the $\Box$ system: $\tctxo; \cso
  \mid \ctxo \vd e : \tyo$ implies $\tctxo; \cso \mid \ctxo,
  \cempty{x}{\tyw} \vd e : \tyo$. Then, using this lemma we can adjust
  the environments so that subtyping goes through in the system with
  $\Box$.
\end{proof}
A $\Box$-elimination operation $\req{R}$, which sends environment
annotations to sensitivities will prove useful in the the syntax
directed system. It is defined as $\req{\Box} = 0$, $\req{R} = R$
otherwise. Remember that $\Box$ doesn't belong to the sensitivity
language, so any annotation that is used in places where a sensitivity
is expected must be wrapped with $\req{\--}$.

\begin{definition}[Extension to environments operations]
  Operations on extended sensitivites that were extended to
  environments in a pointwise fashion, now must take into account the
  presence of $\Box$.
  \begin{itemize}
  \item $\smax{R_1}{R_2}$ operates now as $\smax{\Box}{\Box} = \Box$,
    $\smax{\Box}{R} = R$, $\smax{R}{\Box} = R$, the original term
    otherwise.
  \item $\ssup{i}{R}$ is extended in the natural way $\ssup{i}{\Box} =
    \Box$, the original term otherwise.
  \item $\scase{S}{i}{R_0}{R_s}$ operates now
    $\scase{S}{i}{\Box}{\Box} = \Box$, $\scase{S}{i}{R_0}{R_s} =
    \scase{S}{i}{\req{R_0}}{\req{R_s}}$ otherwise.
  \end{itemize}
\end{definition}

\section{Subtyping Proofs}

From now on we can consider only environments of similar length. We prove
a few necessary facts about subtyping.

\begin{lemma}[Environment manipulation]
  \label{lem:ctx}
  Environment subtyping is preserved by addition and scalar multiplication. More
  formally:
  \begin{itemize}
    \item If $\tctxo; \cso \models \ctxo \tless \ctxo' \land \ctxw \tless
      \ctxw'$, then $\tctxo; \cso \models \ctxo + \ctxw \tless \ctxo' + \ctxw'$;
      and
    \item if $\tctxo; \cso \models \ctxo \tless \ctxo' \land R \geq R'$, then
      $\tctxo; \cso \models R \cdot \ctxo \tless R' \cdot \ctxo'$.
  \end{itemize}
\end{lemma}
\begin{proof}
  These follow from the interpretation of subtyping assertions. Note
  that the subtyping relation preserves the skeleton of the
  environments, thus making sure that the operations are always defined.
\end{proof}

\begin{lemma}[Properties of extended sensitivities] \label{lem:ext}
  Extended sensitivities satisfy the following properties:
  \begin{itemize}
    \item $\tctxo; \cso \models R \geq \smax{R_1}{R_2}$ if and only if
      $\tctxo; \cso \models R \geq R_1 \land R \geq R_2$;
    \item $\tctxo; \cso \models R \geq \ssup{i}{R'}$ with $i \# \tctxo$ if and
      only if $\tctxo, i; \cso \models R \geq R'$; and
    \item $\tctxo; \cso \models R \geq \scase{S}{i}{R_0}{R_s}$ with $i \#
      \tctxo$ if and only if
      \[
        \tctxo; \cso, S = 0 \models R \geq R_0
        \qand
        \tctxo, i; \cso, S = i + 1 \models R \geq R_s.
      \]
  \end{itemize}
  As an immediate corollary, setting $R$ to be $\smax{R_1}{R_2}, \ssup{i}{R'},
  \scase{S}{i}{R_0}{R_s}$ yields
  \begin{itemize}
    \item $\tctxo; \cso \models \smax{R_1}{R_2} \geq R_1 \land R \geq R_2$;
    \item $\tctxo, i; \cso \models \ssup{i}{R'} \geq R'$; and
    \item $\tctxo; \cso, S = 0 \models \scase{S}{i}{R_0}{R_s} \geq R_0$ and
      $\tctxo, i; \cso, S = i + 1 \models \scase{S}{i}{R_0}{R_s} \geq R_s$.
  \end{itemize}
\end{lemma}
\begin{proof}
  These follow from the interpretation of extended sensitivities.
\end{proof}

\begin{lemma} \label{lem:subsub}
  Suppose $\tctxo, i : \kappa; \cso \models \tyo \tless \tyw$ and $i \# \cso$.
  Then for any $\phi \models S : \kappa$, we have
  \[
    \tctxo; \cso \models \tyo[S/i] \tless \tyw[S/i].
  \]
\end{lemma}
\begin{proof}
  By induction on the subtype derivation. For the base cases, we know
  \[
    \forall \tctxo, i : \kappa.\; (\cso \Rightarrow R \geq R'),
  \]
  and we need to prove
  \[
    \forall \tctxo.\; (\cso \Rightarrow R[S/i] \geq R'[S/i]),
  \]
  but this is clear from the interpretation of $R, R'$.
\end{proof}

\section{The Syntax-Directed system}

The syntax-directed system is presented in
\Cref{fig:dfuzz-syn-typing}. It works over a uniform environment, using
$\Box$ annotations to ``mark'', variables not occurring in the
original \DFuzz derivation.

\newcommand{\sdconst}{
  \inferrule
  { }
  { \tctxo;\cso\mid \mathrm{Ectx}(\ctxskel) \vs  \realone : \R }
  \rname{\ruleconstR}
}
\newcommand{\sdvar}{
  \inferrule
  { }
  { \tctxo;\cso\mid \mathrm{Ectx}(\ctxskel), \cass{x}{1}{\tyo} \vs  x : \tyo }                 \rname{\rulevar}
}
\newcommand{\sditens}{
  \inferrule
  { \tctxo;\cso\mid \ctxo_1 \vs e_1 : \tyo \\
    \tctxo;\cso\mid \ctxo_2 \vs e_2 : \tyw }
  { \tctxo;\cso\mid \ctxo_1 + \ctxo_2
    \vs (e_1, e_2) : \tyo \otimes \tyw }                                \rname{\ruleitens}
}
\newcommand{\sdetensold}{
  \inferrule
  { \tctxo;\cso\mid \ctxw \vs e : \tyo_1' \otimes \tyo_2'  \\
    \tctxo;\cso\mid \ctxo, \cass{x}{R_1}{\tyo_1}, \cass{y}{R_2}{\tyo_2} \vs e' :
    \tyw \\\\
    \tctxo; \cso \models \tyo_1' \tless \tyo_1 \land \tyo_2' \tless \tyo_2 }
  { \tctxo;\cso\mid \ctxo + \smax{\req{R_1}}{\req{R_2}} \cdot \ctxw
    \vs \blet (x : \tyo_1, y : \tyo_2) = e \bin e' : \tyw}          \rname{\ruleetens}
}
\newcommand{\sdetens}{
  \inferrule
  { \tctxo;\cso\mid \ctxw \vs e : \tyo \otimes \tyw  \\
    \tctxo;\cso\mid \ctxo, \cass{x}{R_1}{\tyo}, \cass{y}{R_2}{\tyw} \vs e' : \tyt
  }
  { \tctxo;\cso\mid \ctxo + \smax{\req{R_1}}{\req{R_2}} \cdot \ctxw
    \vs \blet (x, y) = e \bin e' : \tyt}          \rname{\ruleetens}
}
\newcommand{\sdiamp}{
  \inferrule
  { \tctxo;\cso\mid \ctxo_1 \vs e_1 : \tyo \\ \tctxo;\cso\mid \ctxo_2 \vs e_2 : \tyw }
  { \tctxo;\cso\mid \smax{\ctxo_1}{\ctxo_2}
    \vs \pair{e_1}{e_2} : \tyo \amp \tyw }                              \rname{\ruleiamp}
}
\newcommand{\sdeamp}{
  \inferrule
  { \tctxo;\cso\mid \ctxo \vs e : \tyo_1 \amp \tyo_2 }
  { \tctxo;\cso\mid \ctxo  \vs \pi_i e : \tyo _i }                      \rname{\ruleeamp}
}
\newcommand{\sdiapp}{
  \inferrule
    { \tctxo;\cso\mid \ctxo, \cass{x}{R^\bullet}{\tyo} \vs  e : \tyw \and
      \tctxo; \cso \models R \geq \req{R^\bullet} }
  { \tctxo;\cso\mid \ctxo
    \vs  \lambda (\cass{x}{R}{\tyo}).~e : !_R \tyo \multimap \tyw }           \rname{\ruleiapp}
}
\newcommand{\sdeapp}{
  \inferrule
  { \tctxo;\cso\mid \ctxo \vs  e_1 : !_R \tyo \multimap \tyw \\\\
    \tctxo;\cso\mid \ctxw \vs  e_2 : \tyo'
    \\ \tctxo;\cso \vsub \tyo' \tless \tyo }
  { \tctxo;\cso\mid \ctxo + R\cdot\ctxw \vs e_1\;e_2 : \tyw }           \rname{\ruleeapp}
}
\newcommand{\sditapp}{
  \inferrule
    { \tctxo, i : \kappa ; \cso \mid \ctxo \vs e : \tyo \\ \text{$i$ fresh in
        $\cso$} }
  {  \tctxo;\cso \mid \ssup{i}{\ctxo}
    \vs \Lambda i :\kappa.\; e : \forall i : \kappa.\; \tyo }           \rname{\ruleitapp}
}
\newcommand{\sdetapp}{
  \inferrule
    { \tctxo ;\cso\mid \ctxo \vs e : \forall i : \kappa.\; \tyo \\ \tctxo \models
    S : \kappa}
  { \tctxo ;\cso \mid \ctxo \vs e[S] : \tyo[S/i] }                      \rname{\ruleetapp}
}
\newcommand{\sdfix}{
  \inferrule
  { \tctxo; \cso \mid \ctxo, \cass{\varone}{R}{\tyo} \vs e : \tyo' \\
    \tctxo; \cso \models \tyo' \tless \tyo }
  { \tctxo; \cso \mid \infty \cdot \ctxo \vs \bfix \varone : \tyo.\; e : \tyo } \rname{\rulefix}
}
\newcommand{\sdnate}{
  \inferrule
  { \tctxo;\cso \mid \ctxw \vs e : \N[S] \\
    \tctxo ;\cso, S = 0     \mid  \ctxo_0 \vs e_0 : \tyo_0 \\\\
    \tctxo, i : \natone ;\cso, S = i + 1 \mid \ctxo_s, \cass{n}{R}{\N[i]} \vs e_s :
    \tyo_s \\\\
    \tctxo; \cso, S = 0 \models \tyo_0 \tless \tyo \\
    \tctxo, i : \natone; \cso, S = i + 1 \models \tyo_s \tless \tyo }
  { \tctxo;\cso\mid \scase{S}{i}{\ctxo_0}{\ctxo_s} +
    \scase{S}{i}{0}{\req{R}} \cdot \ctxw
    \vs \bcase e \breturn \tyo
    \bof 0 \Rightarrow e_0 \mid n_{[i]} + 1 \Rightarrow e_s : \tyo }    \rname{\rulenate}
}
\begin{figure*}
  \centering
\begin{mathpar}
  \sdconst \and
  \sdvar \and
  \sditens \and
  \sdetens \and
  \sdiamp \and
  \sdeamp \and
  \sdiapp \and
  \sdeapp \and
  \sditapp \and
  \sdetapp \and
  \sdfix \and
  \sdnate \and

    \mathrm{Ectx(\ctxskel)} := \ctxw
    \quad \text{with} \quad
    \left\{
      \begin{array}{ll}
        \dom(\ctxskel) &= \dom(\ctxw) \\
        \ctxw(b) &\equiv \cempty{\_}{\_}
        \qquad
        \text{for all } b \in \dom(\ctxskel) \\
      \end{array}
    \right.
\end{mathpar}
  \caption{\DFuzz Type Judgment, Syntax-directed Version}
  \label{fig:dfuzz-syn-typing}
\end{figure*}

We first prove the system sound with respect the non syntax-directed
one.
\begin{lemma}[Syntax-directed soundness] \label{lem:sd-sound}
  If $\tctxo; \cso \mid \ctxo \vs \tmo : \tyo$ has a derivation, then $\tctxo;
  \cso \mid \ctxo \vd e : \tyo$.
\end{lemma}
\begin{proof}
  By induction on the derivation proving $\tctxo; \cso \mid \ctxo \vs
  \tmo : \tyo$.
  \begin{description}
    \item[Case: $\rulevar$]
      \[
        \sdvar
      \]
      Immediate, the same rule applies.
    \item[Case: \ruleitens]
      \[
        \sditens
      \]
      Immediate by induction; the same rule applies.
    \item[Case: \ruleetens]
      \[
        \sdetens
      \]
      By induction, we have
      \[
        \tctxo;\cso\mid \ctxw \vd e : \tyo \otimes \tyw
        \qand
        \tctxo;\cso\mid \ctxo, \cass{x}{R_1}{\tyo}, \cass{y}{R_2}{\tyo} \vd e'
        : \tyt
      \]
      By \Cref{lem:ext}, $\tctxo; \cso \models
      \smax{\req{R_1}}{\req{R_2}} \geq \req{R_i}$ for $i = 1,
      2$. Abbreviating $R^\bullet := \smax{\req{R_1}}{\req{R_2}}$ and
      applying weakening we have:
      \[
        \tctxo;\cso\mid \ctxo, \cass{x}{R^\bullet}{\tyo},\cass{y}{R^\bullet}{\tyw}  \vd e'
        : \tyt
      \]
      with $R^\bullet \neq \Box$ so we have exactly what we need to apply \ruleetens.
    \item[Case: $\ruleiamp$]
      \[
        \sdiamp
      \]
      By induction, we have
      \[
        \tctxo; \cso \mid \ctxo_1 \vd e_1 : \tyo
        \qand
        \tctxo; \cso \mid \ctxo_2 \vd e_2 : \tyw.
      \]
      By \Cref{lem:ext}, we have
      \[
        \tctxo; \cso \models \smax{\ctxo_1}{\ctxo_2} \tless \ctxo_1
        \qand
        \tctxo; \cso \models \smax{\ctxo_1}{\ctxo_2} \tless \ctxo_2.
      \]
      By weakening, we can derive
      \[
        \tctxo; \cso \mid \smax{\ctxo_1}{\ctxo_2} \vd e_1 : \tyo
        \qand
        \tctxo; \cso \mid \smax{\ctxo_1}{\ctxo_2}  \vd e_2 : \tyw,
      \]
      when we can conclude by \ruleiamp.
    \item[Case: \ruleeamp]
      \[
        \sdeamp
      \]
      Immediate; the same rule applies.
    \item[Case: $\ruleiapp$]
      \[
        \sdiapp
      \]
      By induction, we have
      \[
        \tctxo; \cso \mid \ctxo, \cass{x}{R^\bullet}{\tyo} \vd e : \tyw
      \]
      and we know $R \neq \Box$ and:
      \[
        \tctxo; \cso \models R \geq R^\bullet.
      \]
      By weakening, we have
      \[
        \tctxo; \cso \mid \ctxo, x : !_{R} \tyo \vd e : \tyw,
      \]
      and we can conclude by \ruleiapp.
    \item[Case: $\ruleeapp$]
      \[
        \sdeapp
      \]
      By induction, we have
      \[
        \tctxo; \cso \mid \ctxo \vd e_1 : !_R \tyo \lin \tyw
        \qand
        \tctxo; \cso \mid \ctxw \vd e_2 : \tyo'
      \]
      and we also know
      \[
        \tctxo; \cso \models \tyo' \tless \tyo .
      \]
      By subtyping on the right, we can derive
      \[
        \tctxo; \cso \mid \ctxw \vd e_2 : \tyo,
      \]
      and we can conclude with \ruleeapp.
    \item[Case: $\ruleitapp$]
      \[
        \sditapp
      \]
      By induction, we have
      \[
        \tctxo; i : \kappa; \cso \mid \ctxo \vd e : \tyo
      \]
      and $i$ fresh in $\cso$. By \Cref{lem:ext}, we have
      \[
        \tctxo; \cso \models \ssup{i}{\ctxo} \tless \ctxo,
      \]
      and so by weakening, we have
      \[
        \tctxo, i : \kappa; \cso \mid \ssup{i}{\ctxo} \vd e : \tyo.
      \]
      Now, we can conclude with \ruleitapp.
    \item[Case: $\ruleetapp$]
      \[
        \sdetapp
      \]
      Immediate; the same rule applies.
    \item[Case: \rulefix]
      \[
        \sdfix
      \]
      By induction; we have
      \[
        \tctxo; \cso \mid \ctxo, x : !_R \tyo \vd e : \tyo'.
      \]
      But we also have $\tctxo; \cso \models \tyo' \tless \tyo$. By subtyping,
      we get
      \[
        \tctxo; \cso \mid \ctxo, x : !_R \tyo \vd e : \tyo
      \]
      and we can conclude with \rulefix.
    \item[Case: $\rulenate$]
      \[
        \sdnate
      \]
      By induction, we have
      \begin{align*}
        &\tctxo; \cso \mid \ctxw \vd e : \N[S] \\
        &\tctxo; \cso, S = 0 \mid \ctxo_0 \vd e_0 : \tyo_0 \\
        &\tctxo, i : \natone; \cso, S = i + 1 \mid \ctxo_s, n : !_R \N[i] \vd e_s
        : \tyo_s .
      \end{align*}
      By \Cref{lem:ext}, we have
      \begin{align*}
        &\tctxo; \cso, S = 0 \models \scase{S}{i}{\ctxo_0}{\ctxo_s} \tless \ctxo_0
        \\
        &\tctxo, i : \natone; \cso, S = i + 1 \models \scase{S}{i}{\ctxo_0}{\ctxo_s} \tless \ctxo_s
        \\
        &\tctxo, i : \natone; \cso, S = i + 1 \models \scase{S}{i}{0}{\req{R}} \geq \req{R}
      \end{align*}
      with $\req{R} \neq \Box$, and we also know
      \begin{align*}
        &\tctxo; \cso, S = 0 \models \tyo_0 \tless \tyo
        \\
        &\tctxo, i : \natone; \cso, S = i + 1 \models \tyo_s \tless \tyo.
      \end{align*}

      By subtyping on the left and right, we have
      \begin{align*}
        &\tctxo; \cso \mid \ctxw \vd e : \N[S] \\
        &\tctxo; \cso, S = 0 \mid \scase{S}{i}{\ctxo_0}{\ctxo_s} \vd e_0 : \tyo
        \\
        &\tctxo, i : \natone; \cso, S = i + 1 \mid
        \scase{S}{i}{\ctxo_0}{\ctxo_s}, n : !_{R^\bullet} \N[i] \vd e_s : \tyo,
      \end{align*}
      where $R^\bullet = \scase{S}{i}{0}{\req{R}}$.
      We can then conclude by \rulenate.
      \[
      \dfuzzbnate
      \]
  \end{description}
\end{proof}

We now prove completeness, that is to say, for every derivation in the
original system, the syntax-directed one will have a derivation,
possibly even a better from a subtype point of view.

We first need a few auxiliary lemmas:
\begin{lemma} \label{lem:sd-change}
  Suppose that $\tctxo; \cso \mid \ctxo \vs \tmo : \tyo$ is
  derivable. Then, for any logically equivalent $\csw$ such that
  $\tctxo \models \cso \Leftrightarrow \csw$, there is a derivation of
  $\tctxo; \csw \mid \ctxo \vs \tmo : \tyo$ with the same height.
\end{lemma}
\begin{proof}
  By induction on the derivation. The only place the constraint
  environment is used is when checking constraints of the form
  \[
  \tctxo; \cso \models R \geq R'.
  \]
  But since $\csw$ and $\cso$ are logically equivalent, we evidently
  have
  \[
  \tctxo; \csw \models R \geq R'
  \]
  as well.
\end{proof}

\begin{lemma}[Inner Weakening for the Syntax-directed system]
  \label{lem:inner-weakening}
  Assume a derivation $\ctxo, \cass{x}{R}{\tyo} \vs e : \tyw$, a type
  $\tyo'$ such that $\tyo' \tless \tyo$. Then, there exists a type
  $\tyw'$ and a derivation $\ctxo, \cass{x}{R}{\tyo'} \vs e : \tyw'$
  such that $\tyw' \tless \tyw$.
\end{lemma}
\begin{proof}
  By induction over the typing derivation. The base cases are
  immediate. In the induction hypothesis we get to pick the
  appropriate type and we get a better type in all the cases.
\end{proof}

\begin{lemma}[Syntax-directed completeness] \label{lem:sd-comp}
  If $\tctxo; \cso \mid \ctxo \vd \tmo : \tyo$ has a derivation, then there
  exists $\ctxo', \tyo'$ such that $\tctxo; \cso \mid \ctxo'
  \vs \tmo : \tyo'$ has a derivation, $\tctxo; \cso \models \ctxo \tless
  \ctxo'$, $\tctxo; \cso \models \tyo' \tless \tyo$.
\end{lemma}
\begin{proof}
  By induction on the derivation proving $\tctxo; \cso \mid \ctxo \vd \tmo :
  \tyo$.
  \begin{description}
    \item[Case: $\rulesubl$]
      \[
        \dfuzzbsubl
      \]
      Immediate, by induction; the desired environment is $\ctxw$.
    \item[Case: $\rulesubr$]
      \[
        \dfuzzbsubr
      \]
      Immediate, by induction; the desired subtype is $\tyo$.
    \item[Case: $\rulevar$]
      \[
        \dfuzzbvar
      \]
      Immediate; the same rule applies.
    \item[Case: \ruleitens]
      \[
        \dfuzzbitens
      \]
      By induction, we have $\ctxo_1', \ctxo_2', \tyo', \tyw'$ such that
      \[
        \tctxo; \cso \models \ctxo_1 \tless \ctxo_1' \land \ctxo_2 \tless \ctxo_2'
        \qand
        \tctxo; \cso \models \tyo' \tless \tyo \land \tyw' \tless \tyw
      \]
      and derivations
      \[
        \tctxo; \cso \mid \ctxo_1' \vs e_1 : \tyo'
        \qand
        \tctxo; \cso \mid \ctxo_2' \vs e_2 : \tyw'.
      \]
      Then we can conclude by \ruleitens, since \Cref{lem:ctx} shows
      \[
        \tctxo; \cso \models \ctxo_1 + \ctxo_2 \tless \ctxo_1' + \ctxo_2'
        \qand
        \tctxo; \cso \models \tyo' \otimes \tyw' \tless \tyo \otimes \tyw.
      \]
    \item[Case: \ruleetens]
      \[
        \dfuzzbetens
      \]
      By induction and inversion on the subtype relation, we have
      $\ctxw', \ctxo', \tyo', \tyo'', \tyw', \tyw'', \tyt', R_1, R_2$ such that
      \begin{align*}
        &\tctxo; \cso \models \ctxw \tless \ctxw'\\
        &\tctxo; \cso \models \ctxo, \cass{x}{R}{\tyo}, \cass{y}{R}{\tyw}
        \tless \ctxo', \cass{x}{R_1}{\tyo''}, \cass{y}{R_2}{\tyw''} \\
        &\tctxo; \cso \models \tyo' \tless \tyo \land \tyw' \tless \tyw \\
      \end{align*}
      this implies $\tyo' \tless \tyo''$, $\tyw' \tless \tyw''$, $R
      \ge \req{R_1}$, and $R \ge \req{R_2}$.
      We have derivations:
      \[
        \tctxo;\cso\mid \ctxw' \vs e : \tyo' \otimes \tyw'
        \qand
        \tctxo;\cso\mid \ctxo', \cass{x}{R_1}{\tyo''}, \cass{y}{R_2}{\tyw''} \vs e' :
        \tyt'
      \]
      By \Cref{lem:inner-weakening}, we have a derivation:
      \[
        \tctxo;\cso\mid \ctxo', \cass{x}{R_1}{\tyo'}, \cass{y}{R_2}{\tyw'} \vs e' :
        \tyt''
      \]
      with $\tyt'' \tless \tyt'$.
      Hence, we can produce a syntax-directed derivation now:
      \[
      \tctxo; \cso \mid \ctxo' + \smax{\req{R_1'}}{\req{R_2'}} \cdot
      \ctxw' \vs \blet (x, y) = e \bin e' : \tyt''.
      \]

      By \Cref{lem:ext}, we have that $\tctxo; \cso
      \models R \geq \smax{\req{R_1'}}{\req{R_2}}$ and by \Cref{lem:ctx},
      \[
        \tctxo; \cso \models \ctxo + R \cdot \ctxw \tless \ctxo' +
        \smax{\req{R_1'}}{\req{R_2}} \cdot \ctxw',
      \]
      so we are done: the environment $\ctxo' + \smax{\req{R_1'}}{\req{R_2'}} \cdot \ctxw'$ and
      subtype $\tyw''$ suffice.
    \item[Case: $\ruleiamp$]
      \[
        \dfuzzbiamp
      \]
      By induction, there exists
      \begin{align*}
        \tctxo; \cso \models \ctxo \tless \ctxo_1' &\qand
        \tctxo; \cso \models \ctxo \tless \ctxo_2' \\
        \tctxo; \cso \models \tyo' \tless \tyo &\qand
        \tctxo; \cso \models \tyw' \tless \tyw
      \end{align*}
      such that
      \[
        \tctxo; \cso \mid \ctxo_1' \vs e_1 : \tyo'
        \qand
        \tctxo; \cso \mid \ctxo_2' \vs e_2 : \tyw'.
      \]
      By \ruleiamp, we have
      \[
        \tctxo; \cso \mid \smax{\ctxo_1'}{\ctxo_2'} \vs \pair{e_1}{e_2} : \tyo'
        \amp \tyw'.
      \]
      We are done, since by \Cref{lem:ext,lem:ctx},
      \[
        \tctxo; \cso \models \tyo' \amp \tyw' \tless \tyo \amp \tyw
        \qand
        \tctxo; \cso \models \ctxo \tless \smax{\ctxo_1'}{\ctxo_2'} \tless
        \ctxo_i'.
      \]
      So, the desired environment is $\smax{\ctxo_1'}{\ctxo_2'}$, and the desired
      subtype is $\tyo' \amp \tyw'$.
    \item[Case: \ruleeamp]
      \[
        \dfuzzbeamp
      \]
      Immediate, by induction.
    \item[Case: $\ruleiapp$]
      \[
        \dfuzzbiapp
      \]
      By induction, there exists
      \[
        \tctxo; \cso \models \ctxo, \cass{x}{R}{\tyo} \tless \ctxo', x : !_{R'} \tyo
        \qand
        \tctxo; \cso \models \tyw' \tless \tyw
      \]
      such that
      \[
        \tctxo; \cso \mid \ctxo', \cass{x}{R'}{\tyo} \vs e : \tyw'.
      \]
      By inversion on the subtype relation, we have
      \[
        \tctxo; \cso \models R \geq \req{R'} \land \tyw' \tless \tyw.
      \]
      and we are done, since
      \[
        \tctxo; \cso \models !_{\req{R'}} \tyo \lin \tyw' \tless !_R \tyo \lin \tyw
        \qand
        \tctxo; \cso \models \ctxo \tless \ctxo'.
      \]
      \[
      \sdiapp
      \]
    \item[Case: $\ruleeapp$]
      \[
        \dfuzzbeapp
      \]
      By induction, there exists $\ctxo', \ctxw', R', \tyo', \tyw', \tyo''$
      such that
      \begin{align*}
        &\tctxo; \cso \models \ctxo \tless \ctxo' \\
        &\tctxo; \cso \models \ctxw \tless \ctxw' \\
        &\tctxo; \cso \models !_{R'} \tyo' \lin \tyw' \tless !_R \tyo \lin \tyw
        \\
        &\tctxo; \cso \models \tyo'' \tless \tyo,
      \end{align*}
      and derivations
      \[
        \tctxo; \cso \mid \ctxo' \vs e_1 : !_{R'} \tyo' \lin \tyw'
        \qand
        \tctxo; \cso \mid \ctxw' \vs e_2 : \tyo''.
      \]
      By inversion on the subtype relation, we have
      \[
        \tctxo; \cso \models R \geq R'
        \qand
        \tctxo; \cso \models \tyo'' \tless \tyo \tless \tyo'
        \qand
        \tctxo; \cso \models \tyw' \tless \tyw.
      \]
      By \Cref{lem:ext}, the environment $\ctxo' + R' \cdot \ctxw'$ and subtype
      $\tyw'$ suffice.
    \item[Case: $\ruleitapp$]
      \[
        \dfuzzbitapp
      \]
      By induction, there exist
      \[
        \tctxo, i : \kappa; \cso \models \tyo' \tless \tyo
        \qand
        \tctxo, i : \kappa; \cso \models \ctxo \tless \ctxo'
      \]
      such that
      \[
        \tctxo, i : \kappa ; \cso \mid \ctxo' \vs e : \tyo'.
      \]
      Thus, we have the derivation
      \[
        \tctxo; \cso \mid \ssup{i}{\ctxo'} \vs \Lambda i:\kappa .\; e : \forall
        i:\kappa .\; \tyo'
      \]
      and
      \[
        \tctxo; \cso \models \forall i : \kappa .\; \tyo' \tless \forall i :
        \kappa .\; \tyo.
      \]
      By \Cref{lem:ext}, we actually have
      \[
        \tctxo; \cso \models \ctxo \tless  \ssup{i}{\ctxo'} \tless \ctxo',
      \]
      so the environment $\ssup{i}{\ctxo'}$ and subtype $\forall i:\kappa .\;
      \sigma'$ suffices.
    \item[Case: $\ruleetapp$]
      \[
        \dfuzzbetapp
      \]
      By induction, there exists
      \[
        \tctxo; \cso \models \ctxo \tless \ctxo'
        \qand
        \tctxo; \cso \models \forall i:\kappa.\; \tyo' \tless \forall i : \kappa
        .\; \tyo
      \]
      such that
      \[
        \tctxo; \cso \mid \ctxo' \vs e : \forall i:\kappa .\; \tyo'.
      \]
      So, we have a derivation
      \[
        \tctxo; \cso \mid \ctxo'' \vs e[S/i] : \tyo'[S/i].
      \]
      By \Cref{lem:subsub},
      \[
        \tctxo; \cso \models \tyo'[S/i] \tless \tyo[S/i],
      \]
      so the environment $\ctxo'$ and subtype $\tyo'[S/i]$ suffice.
    \item[Case: \rulefix]
      \[
        \dfuzzbfix
      \]
      By induction, we have
      \[
        \tctxo; \cso \models \ctxo, x :!_\infty \tyo \tless \ctxo', x :!_R
        \tyo
        \qand
        \tctxo; \cso \models \tyo' \tless \tyo
      \]
      such that
      \[
        \tctxo; \cso \mid \ctxo', x : !_R \tyo \vs e : \tyo'.
      \]
      We can then conclude by \rulefix: the desired environment is $\infty \cdot
      \ctxo'$ and the desired type is $\tyo$.
    \item[Case: $\rulenate$]
      \[
        \dfuzzbnate
      \]
      By induction, there exists
      \[
        \tctxo; \cso \models \ctxw \tless \ctxw'
        \qand
        \tctxo; \cso \mid \ctxw' \vs e : \N[S']
        \qand
        \tctxo; \cso \models \N[S'] \tless \N[S].
      \]
      By inversion, $\tctxo; \cso \models S = S'$. Also by induction,
      \begin{align*}
        &\tctxo; \cso, S = 0 \models \ctxo \tless \ctxo_0' \\
        &\tctxo, i : \natone; \cso, S = i + 1 \models \ctxo, n : !_R \N[i] \tless
        \ctxo_s', n : !_{R'} \N[i] \\
        &\tctxo; \cso, S = 0 \models \tyo_0' \tless \tyo \\
        &\tctxo, i : \natone; \cso, S = i + 1 \models \tyo_s' \tless \tyo
      \end{align*}
      such that
      \begin{align*}
        &\tctxo; \cso, S = 0 \mid \ctxo_0' \vs e_0 : \tyo_0' \\
        &\tctxo, i : \natone; \cso, S = i + 1 \mid \ctxo_s', n : !_{R'} \N[i] \vs
        e_s : \tyo_s'.
      \end{align*}
      By \Cref{lem:sd-change}, we also have derivations
      \begin{align*}
        &\tctxo; \cso, S' = 0 \mid \ctxo_0' \vs e_0 : \tyo_0' \\
        &\tctxo, i : \natone; \cso, S' = i + 1 \mid \ctxo_s', n : !_{R'} \N[i] \vs
        e_s : \tyo_s'
      \end{align*}
      since $\tctxo; \cso \models S = S'$.

      Hence, we have a derivation
      \begin{align*}
        &\tctxo ; \cso \mid \scase{S'}{i}{\ctxo_0'}{\ctxo_s'} + R^\bullet \cdot
        \ctxw' \\
        &\vs \bcase e \breturn \tyo \bof 0 \Rightarrow e_0 \mid n_{[i]} + 1
        \Rightarrow e_s : \tyo,
      \end{align*}
      where $R^\bullet$ is $\scase{S'}{i}{0}{\req{R'}}$.   We have
      \begin{align*}
        &\tctxo; \cso, S' = 0 \models \scase{S'}{i}{\ctxo_0'}{\ctxo_s'} \tless
        \ctxo_0' \\
        &\tctxo, i : \natone; \cso, S' = i + 1 \models
        \scase{S'}{i}{\ctxo_0'}{\ctxo_s'} \tless \ctxo_s'
      \end{align*}
      so by \Cref{lem:ext}
      \[
        \tctxo; \cso \models \ctxo \tless \scase{S'}{i}{\ctxo_0'}{\ctxo_s'},
      \]
      and
      \[
        \tctxo, i : \natone; \cso, S' = i + 1 \models R \geq R^\bullet \geq \req{R'}
        \qand
        \tctxo, \cso \models R \geq R^\bullet
      \]
      thanks to $R \neq \Box$.

      By weakening, we have
        \begin{align*}
          &\tctxo; \cso \mid \ctxw' \vs e : \N[S'] \\
          &\tctxo; \cso, S = 0 \mid \scase{S'}{i}{\ctxo_0'}{\ctxo_s'} \vs e_0 :
          \tyo \\
          &\tctxo, i : \natone; \cso, S' = i + 1 \mid
          \scase{S'}{i}{\ctxo_0'}{\ctxo_s'}, n : !_{R^\bullet} \N[i]
          \vs e_s : \tyo,
        \end{align*}
      so we can conlude with \rulenate.  The environment
      $\scase{S'}{i}{\ctxo_0'}{\ctxo_s'} + R^\bullet \cdot \ctxw'$ and type
      $\tyo$ suffice (recall that $\tctxo; \cso \models R \geq R^\bullet$, and
      $\tctxo; \cso \models R \cdot \ctxw \tless R^\bullet \cdot \ctxw'$ by
      \Cref{lem:ctx}).
  \end{description}

\end{proof}

\subsection{Algorithm Proofs}
\begin{theorem}[Algorithmic Soundness]
  Suppose $\algin{\tctxo}{\cso}{\ctxskel}{\tmo} \produces
  \algout{\ectxo}{\cctxo}{\ctxo}{\tyo}$. Then, there is a derivation
  of $\tctxo; \cso;\ctxo \vs \tmo : \tyo$.
\end{theorem}
\begin{proof}
  By induction on the algorithmic derivations we see that every
  algorithmic step has an exact correspondence with a syntax-directed
  derivation. We do a few representative cases:
  \begin{description}
  \item[Case \rulevar]
    \[
      \dfuzzaalgvar
    \]

    \[
      \sdvar
    \]
  \item[Case \ruleeapp]
    \[
      \dfuzzaalgeapp
    \]

    \[
      \sdeapp
    \]
  \item[Case \ruleetens]
    \[
      \dfuzzaalgetens
    \]

    \[
      \sdetens
    \]
  \end{description}
\end{proof}
\begin{theorem}[Algorithmic Completeness]
  Suppose
  $\tctxo;\cso; \ctxo \vs \tmo : \tyo$
  is derivable. Then
  $\algin{\tctxo}{\cso}{\ctxskel}{\tmo} \produces
    \algout{\ectxo}{\cctxo}{\ctxo}{\tyo}$.
\end{theorem}
\begin{proof}
  By induction on the syntax-directed derivation. The proof is mostly
  direct, we show a few representative cases.
  \begin{description}
    \item[Case \ruleeapp]
      \[
        \sdeapp
      \]
      By induction, we have derivations
      \[
        \algin{\tctxo}{\cso}{\ctxskel}{\tmo_1}  \produces
        \algout{\ectxo_1}{\cctxo_1}{\ctxo}{!_\sensitermone \tyo \lin \tyw}
        \qand
        \algin{\tctxo}{\cso}{\ctxwskel}{\tmo_2}  \produces
        \algout{\ectxo_2}{\cctxo_2}{\ctxw}{\tyo'}.
      \]
      Note that $\ctxskel = \ctxwskel$ for the syntax-directed derivation to be
      defined, so we can apply the algorithmic rule \ruleeapp:
      \[
        \dfuzzaalgeapp
      \]

    \item[Case \rulefix]
      \[
        \sdfix
      \]
      By induction, we have
      \[
        \algin{\tctxo}{\cso}{\ctxskel, \varone : \tyo}{\tmo}  \produces
        \algout{}{}{\ctxo, \cass{x}{R}{\tyo}}{\tyo'}
      \]
      and we can apply the algorithm rule \rulefix:
      \[
        \dfuzzaalgfix
      \]
    \item[Case \ruleetens]
      \begin{align*}
        \sdetens
      \end{align*}
      We know that $\ctxskel = \ctxwskel$.
      By induction, we know that:
      \begin{align*}
        &\algin{\tctxo}{\cso}{\ctxskel}{\tmo} \produces
        \algout{\ectxo}{\cctxo}{\ctxw}{\tyo_1' \otimes \tyo_2'} \\
        &\algin{\tctxo}{\cso}{\ctxskel, \varone_1 : \tyo_1, \varone_2 : \tyo_2}{\tmo'} \produces
        \algout{\ectxo}{\cctxo}{\ctxo, \cass{x}{R_1}{\tyo_1}, \cass{y}{R_2}{\tyo_2}}{\tyw}
      \end{align*}
      and we know $\tctxo; \cso \models \tyo_1' \tless \tyo_1 \land \tyo_2'
      \tless \tyo_2$, so we apply the algorithmic case \ruleetens:
      \[
        \dfuzzaalgetens
      \]
    \item[Case \rulenate]
      \[
        \sdnate
      \]
      We know that $\ctxskel=\ctxwskel$. By induction, we know that:
      \begin{align*}
        &\algin{\tctxo}{\cso}{\ctxskel}{\tmo} \produces
        \algout{}{}{\ctxw}{\N[\sizeitermone]} \\
        &\algin{\tctxo}{\cso, \sizeitermone=0}{\ctxskel}{\tmo_0} \produces
        \algout{}{}{\ctxo_0}{\tyo_0} \\
        &\algin{\tctxo,\sizeivarone : \natone}{\cso,\sizeitermone=\sizeivarone+1}{\ctxskel,
          \varone : \N[i]}{\tmo_s} \produces
        \algout{}{}{\ctxo_s, \cass{\varone}{R'}{\N[i]}}{\tyo_s}
      \end{align*}
      and we know
      \[
        \tctxo; \cso, S = 0 \models \tyo_0 \tless \tyo
        \qand
        \tctxo, i : \natone; \cso, S = i + 1 \models \tyo_s \tless \tyo.
      \]
      We can conclude with the algorithmic rule \rulenate:
      \[
        \dfuzzaalgcasenat
      \]
  \end{description}
\end{proof}

\section{Minimal Types}

\begin{lemma} \label{ex:minimal}
  \DFuzz does not have minimal types.
\end{lemma}
\begin{proof}
  Using dependent recursion, we can define a function $\mathbf{use} :
  \forall i:\natone.\; !_0 \N[i] \lin !_i \R \lin \R$ that
  multiplies a real number by a natural number. Consider the following
  term $e$:
  \[
    \Lambda i:\natone.\; \lambda e : \N[i], x:\R.\;
    \pair{x}{\mathbf{use}[i]\; e\; x+\mathbf{use}[i]\; e \;x}.
  \]
  Evidently, the minimal type should have the form
  \[
    \emptyset; \emptyset \mid \emptyset \vd  e : \forall i:\natone.\;
    !_0 \N[i] \lin !_{q} \R
    \lin \R \amp \R
  \]
  for some sensitivity expression $q$. What should $q$ be? Note that
  $q(i)$ can be \emph{a priori} a polynomial in $i$ with positive,
  real coefficients. By inspecting the typing rules, we find that
  \[
    i : \natone; \emptyset \models q(i) \geq 1 \land q(i) \geq 2i.
  \]
  Furthermore, the subtyping judgments show that
  \[
    \forall i:\natone.\; !_0 \N[i] \lin !_{a} \R \lin \R \amp \R \tless \forall i:\natone.\;
    !_0 \N[i] \lin !_{b} \R \lin \R \amp \R
  \]
  is equivalent to $i : \natone; \emptyset \models a \leq b$.  Suppose that
  $q(i)$ is the minimal such polynomial for the sensitivity in the type of $e$.
  If the degree of $q$ is strictly greater than $1$, then the polynomial $2i + 1$
  satisfies $2i + 1 \geq 1 \land 2i + 1 \geq 2i$, and is eventually smaller than
  $q$ for large $i$ (since $q$ has higher degree and has non-negative
  coefficients).

  On the other hand, $q$ can't have degree $0$ since it must be larger
  than $2i$ for all $i$. If $q$ has degree $1$, then its leading
  coefficient must be at least 2. Now, the polynomial $i^2 + 1$
  satisfies $i^2 + 1 \geq 1 \land i^2 + 1 \geq 2i$. Finally, note
  \[
    q \geq 2i + 1 \geq i^2 + 1
  \]
  for $i \in \{0, 1\}$. Hence, there is no minimal sensitivity $q$, and hence no
  minimal type for $e$.
\end{proof}

\section{Auxiliary Lemmas}
\begin{lemma}[Standard Annotations]
  %
  Assume annotations in a term $e$ range over regular
  sensitivities and $\tctxo; \cso \mid \ctxo \vs e : \tyo$.
  Then:
  \begin{itemize}
    \item $\tyo$ has no extended sensitivities; and
    \item all the constraints are of the form $\tctxo;\cso\models R
      \geq R'$ where $R$ is a standard sensitivity term.
  \end{itemize}
  This directly implies \Cref{lem:annot}.
\end{lemma}
\begin{proof}
  The first point is clear by inspecting the rules in
  \Cref{fig:dfuzz-syn-typing}: by induction, the type of any
  expression has only regular sensitivities. The second point is also
  clear: in all subtype checks in \Cref{fig:dfuzz-syn-typing}, both
  types have no extended sensitivities by the first point. The only
  place where we check against an extended sensitivity is in rule
  \ruleiapp, with constraint
  \[
    \tctxo;\cso\models R \geq R'.
  \]
  Here, the $R$ is a standard sensitivity term since it is an
  annotation, but the $R'$ may be an extended sensitivity.
\end{proof}

\aa{There were left out of the appendix thanks to some stupid
  mistake... We should make sure they fit nice in here}

\begin{lemma}[Constraint Simplification]
  Suppose $Q \cred Q'$, and suppose $\phi \vdash Q$ and $\phi \vdash Q'$.  Then,
  for any standard valuation $\rho \in\mathsf{val}(\phi)$, we have $\llb Q
  \rrb_\rho = \llb Q' \rrb_\rho$.
\end{lemma}
\begin{proof}
  \aa{What does $\phi_i; \Phi_i\models\rho$ mean?}
  By induction on the derivation of $Q \cred Q'$. The cases $\text{Plus}$,
  $\text{Mult}$ and $\text{Red}$ are immediate by induction. The other cases all
  follow by the semantics of $\bclub$.
  \begin{description}
    \item[Case $\text{Flat}$:]
      The semantics of $Q$ under valuation $\rho$ is equivalent
      to the larger of
      \[
        \max_{\rho'} \{ \max_{i, \rho''} \{ \lb R_i \rb_{\rho \cup \rho' \cup
          \rho''} \st \phi_i; \Phi_i \models \rho' \} \st \phi; \Phi \models
        \rho' \}
      \]
      and $N = \llb \club{V} \rrb_{\rho}$. The first expression can be seen to
      be
      \[
        M = \max_{i, \rho', \rho''} \{ \lb R_i \rb_{\rho \cup \rho' \cup
          \rho''} \st \phi, \phi_i; \Phi \land \Phi_i \models \rho', \rho''
        \},
      \]
      and the semantics of $Q'$ under the valuation can be seen to be
      $\max(M, N)$, as desired.
    \item[Case $\text{CPlus}$:]
      The interpretation of $Q$ under valuation $\rho$ is
      \[
        \max_i \max \{ \lb R_i \rb_{\rho \cup \rho_i} \st \phi_i; \Phi_i
        \models \rho_i \}
        +
        \max_j \max \{ \lb R_j' \rb_{\rho \cup \rho_j'} \st \phi_j'; \Phi_j'
        \models \rho_j' \}
      \]
      The first maximum is achieved at some $i^*$, and the second maximum is
      achieved at $j^*$. Then,
      \[
        \max \{ \lb R_{i^*} \rb_{\rho \cup \rho_i} \st \phi_{i^*}; \Phi_{i^*}
        \models \rho_{i} \}
        +
        \max \{ \lb R_{j^*}' \rb_{\rho \cup \rho_j'} \st \phi_{j^*}';
        \Phi_{j^*}' \models \rho_j' \}
      \]
      is at most
      \[
        \max \{ \lb R_{i^*} + R_{j^*}' \rb_{\rho \cup \rho_i \cup \rho_j'}
        \st \phi_{i^*}, \phi_{j^*}'; \Phi_{i^*} \land \Phi_{j^*}' \models
        \rho_i, \rho_j' \}
        \leq
        \llb \club{ \csens{\phi_i \cup \phi_j'}{\Phi_i \land \Phi_j'}{R_i +
            R_j'} }_{ij} \rrb_\rho
      \]
      since $\phi_{i^*}, \phi_{j^*}$ are assumed to be disjoint. For the reverse
      direction, consider the semantics of $Q'$:
      \[
        \max_{ij} \max \{ \lb R_i + R_j'\rb_{\rho \cup \rho_i \cup \rho_j'}
        \st \phi_i, \phi_j'; \Phi_i \land \Phi_j' \models \rho_i, \rho_j' \}
      \]
      If there are no valuations such that $\phi_i \cup \phi_j';
      \Phi_i \land \Phi_j' \models \rho_i, \rho_j'$, then we are done (we've
      defined the max of an empty set to be $0$). If the maximum is achieved at
      some $\rho, \rho'$ at $i^*, j^*$, then we know $\phi_{i^*}; \Phi_{i^*} \models
      \rho$ and $\phi_{j^*}; \Phi_{j^*} \models \rho'$, when the maximum is at
      most $\llb Q \rrb_\rho$.
    \item[Case $\text{CMult}$:]
      This case follows like the previous case.
  \end{description}
\end{proof}

\end{document}